\documentclass[10pt,leqno]{amsart}
\usepackage{graphicx}
\newcommand{\genlegendre}[4]{%
  \genfrac{(}{)}{}{#1}{#3}{#4}%
  \if\relax\detokenize{#2}\relax\else_{\!#2}\fi
}
\newcommand{\legendre}[3][]{\genlegendre{}{#1}{#2}{#3}}

\newcommand{\N}{\mathbb{N}}

\usepackage[indentafter]{titlesec}
\titleformat{name=\section}{}{\thetitle.}{0.8em}{\centering\scshape }
\titleformat{name=\subsection}{}{\thetitle.}{0.5em}{ \scshape}
\titleformat{name=\subsubsection}[runin]{}{\thetitle.}{0.5em}{\itshape}[.]
\titleformat{name=\paragraph,numberless}[runin]{}{}{0em}{}[.]
\titlespacing{\paragraph}{0em}{0em}{0.5em}
\titleformat{name=\subparagraph,numberless}[runin]{}{}{0em}{}[.]
\titlespacing{\subparagraph}{0em}{0em}{0.5em}

\usepackage{comment}
\usepackage{empheq}
\usepackage{indentfirst,mathrsfs} 
\usepackage{csquotes}
\usepackage{hyphenat}

\usepackage{blindtext}
\usepackage{enumerate}
\usepackage{parskip}

\usepackage{xpatch}

\usepackage{appendix}

\usepackage[noend, vlined, ruled]{algorithm2e}
\usepackage{algpseudocode}

\usepackage[english]{babel}
\usepackage{mathtools}
\usepackage{amstext} 
\usepackage{xcolor}
\usepackage{hyperref}

\usepackage{float}
\usepackage{array}   
\newcolumntype{L}{>{$}l<{$}}

\usepackage{amssymb,amsthm,amsmath}
\usepackage{xcolor,titlesec,fancyhdr,etoolbox}
\newtheorem{theorem}{Theorem}[]

\newtheorem{definition}[theorem]{Definition}

\newtheorem{lemma}[theorem]{Lemma}
\newtheorem{proposition}[theorem]{Proposition}

\hypersetup{ colorlinks=true, linkcolor=black, filecolor=black, urlcolor=black }

\usepackage{lipsum}
\usepackage{float}

\begin{document}

\title{From Worst to Average Case to Incremental Search Bounds of the Strong Lucas Test} 

\author{Semira Einsele \and Gerhard Wunder}
\address{Semira Einsele, Gerhard Wunder; Department of Mathematics and Computer Science, Freie Universit\"at Berlin, Germany  \vspace{-0.3cm}}
\email{semira.einsele@fu-berlin.de}

\kern-3em

\begin{abstract}
The strong Lucas test is a widely used probabilistic primality test in cryptographic libraries. When combined with the Miller-Rabin primality test, it forms the Baillie-PSW primality test, known for its absence of false positives, undermining the relevance of a complete understanding of the strong Lucas test.

In primality testing, the worst-case error probability serves as an upper bound on the likelihood of incorrectly identifying a composite as prime. For the strong Lucas test, this bound is $4/15$ for odd composites, not products of twin primes. 
On the other hand, the average-case error probability indicates the probability that a randomly chosen integer is inaccurately classified as prime by the test. This bound is especially important for practical applications, where we test primes that are randomly generated and not generated by an adversary.

The error probability of $4/15$ does not directly carry over due to the scarcity of primes, and whether this estimate holds has not yet been established in the literature. This paper addresses this  gap by demonstrating that an integer passing $t$ consecutive test rounds, alongside additional standard tests of low computational cost, is indeed prime with a probability greater than $1-(4/15)^t$ for all $t\geq 1$.

Furthermore, we introduce error bounds for the incremental search algorithm based on the strong Lucas test, as there are no established bounds up to date as well. Rather than independent selection, in this approach, the candidate is chosen uniformly at random, with subsequent candidates determined by incrementally adding 2. This modification reduces the need for random bits and enhances the efficiency of trial division computation further.

\keywords{Primality Test \and Strong Lucas Test  \and Average Case Error Probability \and Baillie-PSW Test \and Incremental Search Bound.}
\end{abstract}
\maketitle

\section{Introduction}\label{sec1}
\subsection{Primes and Primality Testing Algorithms in Cryptography}
Other than being mathematically interesting, prime numbers are of great importance in cryptography. Many schemes in public key cryptography rely on choosing certain parameters as primes, exemplified by RSA and the Diffie-Hellmann key exchange protocol. 
These protocols come into play whenever we establish a VPN connection, use secure messaging Apps, or utilize smart cards for contactless payments.
The consequences of prime parameter selection mistakes are potentially catastrophic. It is thus essential to have reliable primality testing algorithms that determine whether a given number is composite or prime with high probability. These algorithms are integral components in nearly every cryptographic library or mathematical software system.

\subsection{Probabilistic and Deterministic Primality Tests}

Until 2004, determining whether deterministic primality tests, capable of consistently distinguishing between primes and composites falls within the $\mathcal{P}$ complexity class without relying on any mathematical conjectures remained unsolved. This issue was ultimately solved in \cite{primesinP} with the introduction of a polynomial time algorithm. Nonetheless, while theoretically sound, they are impractical for everyday use, especially when confronted with the large inputs typical in cryptography, and probabilistic primality tests still outperform deterministic tests in terms of efficiency and remain the preferred choice in most scenarios. This preference persists despite the trade-off involving reduced accuracy. Furthermore, it has been demonstrated that those employed in practice have an error probability that is practically negligible. These tests operate as randomized  Monte Carlo algorithms, ensuring always correct identification of prime numbers while rarely producing false positives. The independent nature of each test round allows for control over the probability of mistakenly classifying a composite as prime by increasing the number of rounds.

In the domain of primality tests, there are two key error probability categories: worst-case and average-case error probabilities. The \textit{worst-case} error probability is the maximum probability that a composite integer will be mistakenly identified as a prime number by the test.
Conversely, the \textit{average-case} error probability considers the probability of a random number being misidentified as a prime number by the test. 
In adversarial scenarios, such as in the Diffie-Hellman key exchange protocol, the parameters could be chosen by an adversary. They might intentionally construct composites with a higher likelihood of being wrongly declared prime by the probabilistic primality test compared to randomly selected numbers. Therefore, it is essential for the primality test to exhibit a small worst-case error probability. 
Conversely, in numerous other applications, understanding how the test performs in the average-case is more important. Typically, most randomly chosen composites are less likely to be accepted compared to those with the highest-probability of fooling the primality test.

\subsection{The Miller-Rabin, strong Lucas and Bailllie-PSW Test}
The Miller-Rabin primality test is frequently utilized in cryptographic libraries and mathematical software due to its relatively straightforward implementation, efficient running time, and well-established error bounds, which have garnered trust within the cryptographic community. The work in \cite{DamEtAl} provides average-case error estimates, while the works in \cite{Rabin}, \cite{Monier} contribute worst-case error estimates. Detailed mathematical information about the test is available in  Appendix \ref{fermat} for those interested.

Another notable probabilistic primality test is the Lucas test and its more stringent variant, the strong Lucas test. Similar to the Miller-Rabin test, there exist both worst-case \cite{Rabin-Mon-Lucas} and average-case error estimates \cite{Einseleav}.
Besides being one of the main primality tests implemented in cryptolibraries, it gains importance due to its role in the Baillie-PSW test, which is a specific combination of the Miller-Rabin and Lucas test, see Algorithm \ref{Baillie-PSW-alg}. So far, no false positives have been identified passing this combined test, and the challenge of constructing a single concrete example remains an unsolved problem. Indeed, Gilchrist \cite{gilchrist} computed the number of Baillie-PSW pseudoprimes up to $2^{64}$ and showed that there are none. Empirical data suggest that it seems very unlikely that integers of cryptographic size would be Baillie-PSW pseudoprimes. 
However, Pomerance \cite{pomerance1984there} presents a heuristic argument positing the existence of infinitely many Baillie-PSW pseudoprimes. 
It is worth noting that the parameter selection of the test as implemented in practice is deterministic, hence, the result will remain constant. This can also be seen as an advantage since it eliminates the necessity for randomness.  This reasoning undermines the significance of fully understanding the strong Lucas test.

The interested reader can refer to Appendix \ref{fermat} for details on the Fermat, Miller-Rabin and Lucas test. Appendix \ref{baillietest} provides the algorithm for the Baillie-PSW test. Appendix \ref{orthogonality} provides a heuristic argument for their apparent independence. Lastly, Appendix \ref{rationale} discusses parameter selection for independence between the (strong) Lucas test and the Fermat/Miller-Rabin test.

\subsection{The Algorithm using the Strong Lucas test}\label{algorithm}
Let us introduce at following algorithm, which will be the main part of the first sections. It is a primality testing algorithm, based on the strong Lucas test, which will be discussed in Subsection \ref{strongtest}.

\begin{algorithm}
\caption{\textsc{StrongLuc}($t, k$)}
\label{algorithmA}

\textbf{Input:} $D \in \mathbb{N}$\\
\textbf{Output:} First probable prime found 
\begin{enumerate}
    \item  Choose an odd $k$-bit integer $n$ uniformly at random
    \item If $\legendre{D}{n} \neq -1$, discard $n$ and go to step 1
    \item If ($\gcd(D,n)=1$) or ($n$ is divisible by any of the first 8 odd primes) or (Newton's method finds a square root for $n-\legendre{D}{n}$): discard $n$, go to step 1
    \item Else, execute the following loop: \\
    For $i=1$ to $t$:
    \begin{itemize}
        \item Perform the strong Lucas test to $n$ with randomly chosen bases
       \item If $n$ fails any round of the test, discard $n$ and go to step 1
       \item Else output $n$ and \textbf{stop} 
    \end{itemize}
\end{enumerate}
\end{algorithm}

In the context of algorithm \textsc{StrongLuc}($t, k$), let us define the following quantity.

\begin{definition}\label{qkt}
 Let $q_{k,t}$ represent the probability that an integer that is selected by algorithm \textsc{StrongLuc}($t, k$)  is composite. Here $k$ represents the bit size of the integer, and $t$ corresponds to the number of independent rounds conducted in the strong Lucas test. The average-case behaviour of the algorithm can then be defined as:
\begin{equation}
q_{k,t} = \mathbb{P}(\textnormal{\textsc{StrongLuc}($t, k$)} \text{ outputs a composite}).    
\end{equation}   
\end{definition}

It is noteworthy that OpenSSL already incorporates the practice of dividing by small primes before running the Miller-Rabin test to speed up prime generation. Hence, attaining the error bound of the algorithm, which includes trial division, typically incurs no additional computational costs. Furthermore, the occurrence of twin-prime products can also easily be avoided, as discussed in Subsection \ref{worstcase}. Therefore, achieving this error probability is not necessarily associated with an increase in running time.

\subsection{Contributions}

Arnault provided in \cite{Rabin-Mon-Lucas}  worst-case estimates of the Strong Lucas test, that is, the maximal probability that a composite number, not a product of twin primes, relatively prime to $D$ and distinct from $9$, will falsely be declared prime at most $4/15$-th the time. So any composite passes $t$ independent rounds of the strong Lucas tests with a probability less than or equal to $(4/15)^{t}.$ 

While it may seem logical to directly deduce that $q_{k,t} \leq \big( \frac{4}{15}\big)^t$ from this result, such a conclusion is wrong, as demonstrated in Section \ref{preliminaries}. The question of whether this estimate holds has not been answered yet. 
\newline 

The first contribution of this paper is to close this gap by proving that the bound $q_{k,t}\leq \big( \frac{4}{15}\big)^t$ is true for all $k\geq 2$ and $t\geq 1$.

In order to do so, we use a method introduced in \cite{Burthe} established for the Miller-Rabin test and adapt it for the strong Lucas test. For $k\geq 101$, we show that our claim is a trivial consequence of the average-case error results from \cite{Einseleav}. In order to prove it for smaller values of $k$, we need to extend some of the results in \cite{Einseleav}, enabling us to confirm that $q_{k,t}\leq \big( \frac{4}{15}\big)^t$ for $k\geq 17$ and $t\geq 1$. For the values $2 \leq k \leq 16$, we compute $q_{k,t}$ exactly, which proves our claim.
\newline

The second contribution of this paper is the derivation of error bounds for an adapted algorithm known as  ``incremental search''. Rather than uniformly selecting each candidate at random, the initial candidate is chosen using this approach, while all subsequent candidates are generated by incrementally adding 2. This approach provides the benefit of both conserving random bits and enhancing the efficiency of the trial division calculations. While similar bounds exist for the Miller-Rabin test \cite{incrementalBrandt}, they have not been established for the strong Lucas test. Our work addresses this gap by providing the necessary error bounds.

\section{Preliminaries}\label{preliminaries}

\subsection{Lucas sequences}
Let $D, Q \in \mathbb{Z}$, and $P\in \mathbb{N}$ such that $D=P^2-4Q \geq 0$. Let $U_0(P,Q)=0$, $U_1(P,Q)=1$, $V_0(P,Q)=2$ and $V_1(P,Q)=P.$
The Lucas sequences $U_n(P, Q)$ and $V_n(P,Q)$ associated with the parameters $P$, $Q$ are defined recursively for $n \geq 2$ by

\begin{align}
\label{eqn:Lucassequences}
\begin{split}
    U_n(P,Q)=PU_{n-1}(P,Q)-QU_{n-2}(P,Q),\\
    V_n(P,Q)=PV_{n-1}(P,Q)-QV_{n-2}(P,Q).
\end{split}
\end{align}

\subsection{The Lucas test}
For fixed $D\in \mathbb{Z}$ and $n\in \mathbb{N}$, let $\epsilon_D(n)$ denote the Jacobi symbol $\legendre{D}{n}$. The following theorem is a more relaxed variant of our main theorem, which will be introduced in the next subsection.

\begin{theorem}[Baillie, Wagstaff \cite{Lucas-Baillie}]
Let $P$ and $Q$ be integers, and $D= P^2-4Q$. 
 Let $U_p(P,Q)$ be the Lucas sequence of the first kind. If $p$ is an odd prime such that $(p,QD)=1$, then the following congruence  holds
\begin{equation}\label{lucascong}
    U_{p-\epsilon_D(p)}\equiv 0 \bmod p.
\end{equation}
\end{theorem}

This theorem serves as a basis for the so called \textit{Lucas (primality) test} which checks the congruence (\ref{lucascong}) for several randomly chosen bases $P$ and $Q$. Composites $p$  satisfying  this congruence are called \textit{Lucas pseudoprimes with parameters $P$ and $Q$,} short \textit{lpsp$(P,Q)$} However, similar to Carmichael numbers, which are composites that always pass the primality test based on Fermat's little theorem, a weaker variant of the Miller-Rabin theorem, there are composites, that completely defeat the Lucas test.

\begin{definition}[Lucas-Carmichael numbers]
    Let D be a fixed integer. A Lucas-Carmichael number is a composite number $n$, relatively prime to $2D$, such that for all integers $P, Q$ with $\gcd(P,Q)=1$, $D=P^2-4Q$ and $\gcd(n,QD)=1$, $n$ is a $lpsp(P,Q)$.
\end{definition}

The following theorem further highlights their resemblance to Carmichael numbers.
\begin{theorem}[Williams \cite{williams}]
    Let $D$ be a fixed integer. Then $n$ is a \textit{lpsp$(P,Q)$} if and only if $n$ is square-free and $\epsilon_D(p_i)-1 \mid \epsilon_D(n)-1$ for every prime $p_i \mid n$.
\end{theorem}

In fact, if $n$ is a Lucas-Carmichael number with $D=1$ or $D$ being a perfect square, then $n$ is a Carmichael number. Hence, any results about the infinitude of Lucas-Carmichael numbers, which is still an open question, would generalize the findings concerning Carmichael numbers \cite{alford1994there}, a result that itself took 84 years to prove.

\subsection{The strong Lucas test}\label{strongtest}
Since the Lucas test can never detect Lucas-Carmichael numbers as composites, slight modifications to the test can eliminate this misidentification.
The following theorem will serve as the basis of the strong Lucas (primality) test.
\begin{theorem}[Baillie, Wagstaff \cite{Lucas-Baillie}]\label{defslpsp}
Let $P$ and $Q$ be integers, and $D= P^2-4Q$. Let $p$ be a prime number not dividing $2QD$. Write $p-\epsilon_D(p)=2^\kappa q$, where $q$ is odd. Then
\begin{equation}\label{thrm-slpsp} \text{ either } p \mid U_q 
    \text{ or } 
     p \mid V_{2^iq} \text{ for some $0\leq i <\kappa$. }
\end{equation}
\end{theorem}

By checking property (\ref{thrm-slpsp}) for many uniformly at random chosen bases $P,Q$ with $1\leq P,Q \leq n$, $\gcd
(Q,n)=1$ and $P=D^2-4Q$, we obtain a primality test called \textit{the strong Lucas test.}

In Algorithm \ref{algorithmA}, introduced in Subsection \ref{algorithm}, we only considered integers for which $\epsilon_D(n)=-1$. The rationale behind this choice is explained in Appendix \ref{rationale}.
\subsection{Strong Lucas pseudoprimes}

While Theorem \ref{defslpsp} is generally not applicable to composites, there exist specific bases $P$ and $Q$ for which the theorem holds.
We call a composite number $n$ relatively prime to $2QD$ that satisfies congruence (\ref{thrm-slpsp}) \textit{a strong Lucas pseudoprime} with respect to the parameters $P$ and $Q$, for short $slpsp(P,Q)$. 
Any $slpsp(P,Q)$ is also a $lpsp(P,Q)$ for the same parameter pair, but the converse is not necessarily true, see \cite{Lucas-Baillie}. Therefore, the \textit{strong} Lucas  test is a more stringent test for primality.

\begin{definition} Let $D$ and $n$ be fixed integers and let $n-\epsilon(n)=2^\kappa q$.
We define $SL(D,n)$ to denote the number of pairs $P,Q$ with $0 \leq P,Q <n$, $\gcd(Q,n)=1$, and $P^2-4Q \equiv D \bmod n$, such that $n$ satisfies (\ref{thrm-slpsp}). 
For an integer $n$, which is not relatively prime to $2D$, we set $SL(D,n)=0$. 
\end{definition}

If $n$ is composite, then $SL(D,n)$ counts the number of pairs $P,Q$ that make $n$ a $slpsp(P,Q)$. If $n$ is prime, then $SL(D,n)=n-1-\epsilon(n)$.

From now on, let $D$ always denote a random but fixed integer. For $n$ such that $\gcd(D,2n)=1$, Arnault \cite{Rabin-Mon-Lucas} gave an exact formula on how many pairs $P,Q$ with $0 \leq P,Q <n$, $\gcd(Q,n)=1$, and $P^2-4Q \equiv D \bmod n$ exist that make $n$ a $slpsp(P,Q)$ if we know the prime decomposition of $n.$

\begin{theorem}[Arnault \cite{Rabin-Mon-Lucas}]\label{SL(D,n)}
Let $D$ be an integer and $n=p_1^{r_1}\ldots p_s^{r_s}$ be the prime decomposition of an integer $n\geq 2$ relatively prime to $2D$. Put
\begin{equation*}
 \begin{cases} n- \epsilon_D(n)=2^\kappa q 
    \\ p_i-\epsilon_D(p_i)=2^{k_i}q_i \text{ for } 1\leq i \leq s \end{cases} \text{ with } q, q_i \text{ odd},
    \end{equation*}
ordering the $p_i$'s such that $k_1 \leq \ldots \leq k_s$. The number of pairs $P, Q$ with $0 \leq P, Q < n$, $\gcd(Q,n)=1$, $P^2-4Q \equiv D \bmod n$ and such that $n$ is an \textnormal{slpsp}$(P,Q)$ is expressed by the formula
\begin{equation}\label{SLDn}
    SL(D,n)= \prod_{i=1}^s (\gcd(q,q_i)-1) + \sum_{j=0}^{k_1-1}2^{js}\prod_{i=1}^s\gcd(q,q_i).
\end{equation}
\end{theorem}

\subsection{Worst-case error estimates}\label{worstcase}

Arnault also gave worst-case error estimates.
\begin{theorem}[Arnault \cite{Rabin-Mon-Lucas}]\label{Rabin-Monier-Lucas-using-n}
For every integer $D$ and composite number $n$ relatively prime to $2D$ and distinct from 9, we have
$$
SL(D,n) \leq \frac{4n}{15},
$$
except when $n$ is the product of twin-primes. In this case we have $SL(D,n) \leq n/2.$
\end{theorem}

For certain types of twin-prime products, half of the bases $P,Q$  declare the integer as a prime. Fortunately, excluding all twin-prime products from consideration does not impose a significant restriction. Whenever $n=p(p+2)$ with $\epsilon_D(n)=-1$ and $p$ prime, we can rewrite $n-\epsilon_D(n)=(p+1)^2$, which is a perfect square. This can be efficiently detected, by for example, implementing Newton's method for square roots, as detailed in Algorithm \ref{newton} in the appendix, as a subroutine before executing the actual strong Lucas test.

\subsection{Why worst-case estimates do not imply average-case estimates}

Transitioning from the worst-case to the average-case scenario, our algorithm guarantees the exclusion of integers divisible by the first 
$l$ odd primes and those forming twin-prime products. This perspective enables us to characterize the algorithm as randomly sampling from a set that avoids these specific numbers. This observation leads to the following definition.
\begin{definition}
For $k,l \in \mathbb{N}$ with $k\geq 2$ and, let $M_{k,l}$ denote the set of odd $k$-bit integers that are neither twin-prime products nor divisible by the first $l$ odd primes.  
\end{definition}

With this notation, we can express that the process outlined in Algorithm \ref{algorithmA} effectively corresponds to a uniform sampling from the set $M_{k,8}$.

One might think that from Theorem \ref{Rabin-Monier-Lucas-using-n} it would immediately follow that $q_{k,t} \leq \big(\frac{4}{15}\big)^{t}$, where $q_{k,t}$ was introduced in Definition \ref{qkt}. However, this reasoning is wrong since it does not take into account the distribution of primes, as the following discussion manifests.

Let $X$ be the event that a number chosen uniformly at random from $M_{k,l}$ is composite, and let $E_i$ denote the event that an integer chosen uniformly at random from $M_{k,l}$ passes the $i$-th round of the strong Lucas test. Moreover, let $Y_t$ denote the event that this integer passes $t$ consecutive rounds of the strong Lucas test with uniformly chosen bases. Hence, $Y_t= E_1 \cap E_2 \cap \dots \cap E_t$. Theorem \ref{Rabin-Monier-Lucas-using-n} states that $\mathbb{P}[Y_t\mid X] \leq (\frac{4}{15})^t.$ Critical to the estimation of $q_{k,t}$ is the value of $\mathbb{P}[X \mid Y_t],$ given that $q_{k,t}=\mathbb{P}[X \mid Y_t]$. Naturally, we have $\mathbb{P}[Y_t] \geq \mathbb{P}[X^\mathsf{c}].$ Then, by Bayes' Theorem, we have
\begin{align*}
\mathbb{P}[X \mid Y_t] &= \frac{\mathbb{P}[X]\mathbb{P}[Y_t \mid X]}{\mathbb{P}[Y_t]} \leq \frac{\mathbb{P}[Y_t \mid X]}{\mathbb{P}[Y_t]} = \frac{\mathbb{P}[E_1 \cap \dots \cap E_t \mid X]}{\mathbb{P}[Y_t]} \\
&= \frac{1}{\mathbb{P}[Y_t]}\prod_{i=1}^t \mathbb{P}[E_i \mid X] \leq\frac{1}{\mathbb{P}[Y_t]} \Big(\frac{4}{15} \Big)^t \leq \frac{1}{\mathbb{P}[X^\mathsf{c}]} \Big(\frac{4}{15} \Big)^t,
\end{align*}
where $X^\mathsf{c}$ denotes the complement of $X$.

We generally assume that primes in $M_{k,l}$ are scarce, which means that $\mathbb{P}[X^\mathsf{c}]$ is small. This would imply that our estimate for $\mathbb{P}[X \mid Y_t]$ may be considerably larger than $(\frac{4}{15})^t$ and close to $1$. However, intuitively, $q_{k,t}$ is small, and in \cite{Einseleav} explicit upper bounds have been established, confirming that in fact it is small. So this approach cannot be used to claim that $q_{k,t}\leq (4/15)^t$ for $k\geq 2$ and $t \geq 1$. In this work, however, we indeed show that this bound holds. We demonstrate this by establishing that the probability of an integer passing a round of the test, denoted as $\mathbb{P}[E_i \mid X]$, is significantly lower than $4/15$.

\section{Proof Strategy}

In this section, we outline our approach to proving the first main result. 
Let us introduce the following quantity, which will be used frequently:

\begin{definition}
For $n\in \mathbb{N}$, let $${\overline{\alpha}_D(n)}=\frac{SL(D,n)}{n-\epsilon_D(n)-1}$$ be the proportion of number of pairs $P,Q$ that declare $n$ to be a strong Lucas pseudoprime. 
\end{definition}
Let us first establish some important lemmas.
\subsection{Important Lemma and its consequences}

To prove the first main result, we use the next lemma, adapted from \cite{Burthe} for the strong Lucas test, where we choose $r$ accordingly. For ease of notation, let $\sum^{'}$ denote the sum over composites.
\begin{lemma}\label{4/15-estimate}
Let $r, t , k\in \mathbb{N}$ with $ r < t$ and $k\geq 2.$ Then 
$$
q_{k,t} \leq \Big(\frac{4}{15} \Big)^{t-r} \frac{q_{k,r}}{1-q_{k,r}}.
$$
\end{lemma}
\begin{proof}
We follow Burthe's \cite{Burthe} proof step by step with adequate adaptations for the strong Lucas test. We include the proof for the sake of completeness.
For every $n\in M_{k,l}$, we have $\overline{\alpha}_D(n)\leq 4/15$, since twin-prime products by definition do not belong to $M_{k,l}.$ Moreover, we have $\mathbb{P}[X \cap Y_i]= \frac{1}{\lvert M_{k,l} \rvert} \sum_{n \in M_{k,l}}^{'}\overline{\alpha}_D(n)^{i}.$
For $r < t$, we have
\begin{align*}
    q_{k,t} \hspace{-0.5mm}= \hspace{-0.5mm}\mathbb{P}[X \mid Y_t] \hspace{-0.5mm}= \hspace{-0.5mm}\frac{\mathbb{P}[X \cap Y_t]}{\mathbb{P}[Y_t]}\hspace{-0.5mm}= \frac{\mathbb{P}[X \cap Y_t]}{\mathbb{P}[X \cap Y_{t-1}]}\frac{\mathbb{P}[X \cap Y_{t-1}]}{\mathbb{P} [X \cap Y_{t-2}]}\dots \hspace{-1mm}\frac{\mathbb{P}[X \cap Y_{r+1}]}{\mathbb{P}[X \cap Y_{r}]}\frac{\mathbb{P}[X \cap Y_r]}{\mathbb{P}[Y_{t}]}.
\end{align*}
We bound the fractions as follows:
$$
\frac{\mathbb{P}[X \cap Y_i]}{\mathbb{P}[X \cap Y_{i-1}]}= \frac{\sideset{}{'}\sum\limits_{n \in  M_{k,l}}\overline{\alpha}_D(n)^{i}}{\sideset{}{'}\sum\limits_{n \in  M_{k,l}}\overline{\alpha}_D(n)^{i-1}} \leq \frac{\sideset{}{'}\sum\limits_{n \in M_{k,l}}\frac{4}{15}\overline{\alpha}_D(n)^{i-1}}{\sideset{}{'}\sum\limits_{n \in M_{k,l}}\overline{\alpha}_D(n)^{i-1}}=\frac{4}{15}.
$$
This implies
$$
q_{k,t} \leq \Big( \frac{4}{15} \Big)^{t-r}\frac{\mathbb{P}[X \cap Y_r]}{\mathbb{P}[Y_r]}\frac{\mathbb{P}[Y_r]}{\mathbb{P}[Y_t]} = \Big( \frac{4}{15} \Big)^{t-r} q_{k,r} \frac{\mathbb{P}[Y_r]}{\mathbb{P}[Y_t]}.
$$
Primes in $M_{k,l}$ always pass the strong Lucas test, thus we have that $\mathbb{P}[X^\mathsf{c} \cap Y_t]= \mathbb{P}[X^\mathsf{c}]=\mathbb{P}[X^\mathsf{c} \cap Y_r]$. Therefore,
$$
\frac{\mathbb{P}[Y_r]}{\mathbb{P}[Y_t]}\leq \frac{\mathbb{P}[Y_r]}{\mathbb{P}[X^{\mathsf{c}}\cap Y_t]}=\frac{\mathbb{P}[Y_r]}{\mathbb{P}[X^{\mathsf{c}}\cap Y_r]}=\frac{1}{\mathbb{P}[X^{\mathsf{c}}\mid Y_r]}=\frac{1}{1-q_{k,r}},
$$
which completes the proof.
\end{proof}
Thus, to establish $q_{k,t} \leq \Big(\frac{4}{15}\Big)^t$, it is sufficient to show that for any given $r,k \in \mathbb{N}$ and $k\geq 2$, we have that
$$q_{k,r}\leq\frac{1}{1+\Big (\frac{15}{4}\Big)^r} \; \; \textnormal{ and } q_{k,r'}\leq  \Big(\frac{4}{15}\Big)^{r'} \textnormal{ for all } r'<r,$$
since, by utilizing Lemma \ref{4/15-estimate}, we can then conclude $q_{k,t} \leq  \Big(\frac{4}{15} \Big)^{t}$ for all $t,k\in \mathbb{N}$ and $k\geq 2$.\\
Now, let $\pi(x)$ denote the prime counting function up to $x$, and let $p$ always denote a prime. Using the law of conditional probability, we have 
\begin{align*}
&q_{k,r} = \mathbb{P}[X \mid Y_r]= \frac{\mathbb{P}[X \cap Y_r]}{\mathbb{P}[Y_r]} = \frac{\sideset{}{'}\sum\limits_{n\in M_{k,l}}\overline{\alpha}_D(n)^r}{\sideset{}{}\sum\limits_{n\in M_{k,l}}\overline{\alpha}_D(n)^r} \\
&=\frac{\sideset{}{'}\sum\limits_{n\in M_{k,l}}\overline{\alpha}_D(n)^r}{\sideset{}{'}\sum\limits_{n\in M_{k,l}}\overline{\alpha}_D(n)^r+\sideset{}{}\sum\limits_{p\in M_{k,l}} 1} = \frac{\sideset{}{'}\sum\limits_{n\in M_{k,l}}\overline{\alpha}_D(n)^r}{\sideset{}{'}\sum\limits_{n\in M_{k,l}}\overline{\alpha}_D(n)^r+\pi(2^k)-\pi(2^{k-1})}.
\end{align*}
Given that the expression $\frac{x}{x+\pi(2^k)-\pi(2^{k-1})}$is a monotonically increasing function in $x$, our goal is to find suitable values $N_r$ and $P$ such that $$\sideset{}{'}\sum\limits_{n\in M_{k,l}}\overline{\alpha}_D(n)^r\leq N_r$$ 
and 
$$P \leq \pi(2^k)-\pi(2^{k-1})$$ to establish a bound for $q_{k,r}$ . With these choices, we can express the bound for $q_{k,r}$ as:
\begin{equation}\label{NrP}
q_{k,r} \leq \frac{N_r}{N_r+P}.
\end{equation}
For a fixed $r\in \mathbb{N}$, our aim is to demonstrate that  $$\frac{N_r}{N_r+P} \leq \frac{1}{1+ \Big(\frac{15}{4}\Big)^r}.$$

\subsection{Lower bound \texorpdfstring{$P$}{P}}

The following result serves as our lower bound $P$:
\begin{proposition}\label{k-bit-prime-approx}
For an integer $k\geq 8$, we have
\begin{equation*}
    \pi(2^k)-\pi(2^{k-1}) > (0.71867)\frac{2^k}{k}.
\end{equation*}
\end{proposition}
\begin{proof}
For $k\geq 21$ this is proven in \cite{DamEtAl}. By running a Python program that computes the actual value of $\pi(2^k)-\pi(2^{k-1})$ for $k\leq 20$, we can in fact see the proposition is true for all $k\geq 8$, as can be seen in Table \ref{estimateprimes}.

\begin{table}[!htbp]
\begin{minipage}{.45\linewidth}
\centering
\begin{tabular*}{\linewidth}{@{\extracolsep{\fill}}  L | L  L }
	k &  \pi(2^k)-\pi(2^{k-1}) & \lfloor (0.71867)\frac{2^k}{k} \rfloor \\[-3mm]\\ \hline 
	8 & 23 & 22 \\ 
	9 & 43 & 40\\ 
	10 & 75 & 73 \\ 
	11 & 137 & 133\\ 
	12 & 255 & 245 \\ 
    13 & 464 & 452 \\
    14 & 872 & 841 
\end{tabular*}
\end{minipage}
\hspace{1cm}
\begin{minipage}{.45\linewidth}
\centering
\begin{tabular*}{\linewidth}{@{\extracolsep{\fill}}  L | L L  }
	k &  \pi(2^k)-\pi(2^{k-1}) & \lfloor (0.71867)\frac{2^k}{k} \rfloor \\[-3mm]\\ \hline
	15 & 1612 & 1569 \\ 
	16 & 3030 & 2943\\ 
	17 & 5709 & 5541 \\ 
	18 & 10749 & 10466 \\ 
	19 & 20390 & 19831 \\ 
    20 & 38635 & 37679 \\
    & & 
\end{tabular*}
\end{minipage}
\caption{The exact values of $\pi(2^k)-\pi(2^{k-1})$ and lower bound estimates $\lfloor (0.71867)\frac{2^k}{k}\rfloor $.}
\centering
\label{estimateprimes}
\end{table}
\end{proof}

\newpage

Let us now turn to the upper bound $N_r.$
\subsection{Upper bound \texorpdfstring{$N_r$}{Nr}}

We aim to find an upper bound $N_r$ for $\sideset{}{'}\sum\limits_{n\in M_{k,l}}\overline{\alpha}_D(n)^r.$ First, let us introduce the following set:

\begin{definition}
Let $${C_{m,D}=\{n \in \mathbb{N} : \gcd(n,2D)=1, n \text{ composite and }\alpha_D(n) > 2^{-m}\}}.$$  
\end{definition}
 Arnault proved in \cite{doct-thesis-arnualt} that $\alpha_D(n)\leq \frac{1}{4}$, so we know that $C_{1,D}= C_{2,D} = \emptyset$. The set $C_{3,D}$ has been classified in \cite{Einseleav}.
 
To get our upper bound $N_r$, we use a number-theoretic function introduced by Arnault. This function serves as a variant of the well-known Euler's totient function $\varphi(n)$, and will play a crucial role in the subsequent analysis.

\begin{definition}[Arnault \cite{doct-thesis-arnualt}]
Let $D$ be an integer. The following function is defined only on integers relatively prime to $2D$:
\begin{equation*}
    \begin{cases}\varphi_D(p^r)=p^{r-1}(p-\epsilon_D(p)) \textnormal{ for any prime } p \nmid 2D \textnormal{ and } r\in\mathbb{N},
    \\ \varphi_D(p_1 p_2)= \varphi_D(p_1) \varphi_D(p_2) \text{ if } \gcd(p_1,p_2)=1. \end{cases}
\end{equation*}
Moreover, for odd $n\in \mathbb{N}$, let $$\alpha_D(n)=\frac{SL(D,n)}{\varphi_D(n)}.$$
\end{definition}
In our analysis, we seek to bound $\overline{\alpha}_D(n)$. However, unlike $\varphi(n)$, $\varphi_D(n)$ is not bounded by $n$. Consequently, we cannot straightforwardly establish that $\overline{\alpha}_D(n) \leq \alpha_D(n)$. Nonetheless, the next lemma offers an upper bound and implies that, for sufficiently large $l$, these two quantities become closely aligned. Before looking into the lemma, let us introduce the following definition:
\begin{definition}
For $l\in \mathbb{N}$, let $\tilde{p_l}$ denote the $l$-th odd prime and $\rho_l = 1+ \frac{1}{\tilde{p}_{l+1}}$.   
\end{definition}

Now we can state the lemma, which connects the two functions.

\begin{lemma}[Einsele, Paterson \cite{Einseleav}]\label{newbound}
Let $k,m,l \in \mathbb{N}$ and $n\in C_{m,D} \cap M_{k,l}$ be relatively prime to $2D$. Then 
\begin{equation*}
    \overline{\alpha}_D(n) \leq \rho_l^m \alpha_D(n).
\end{equation*}
\end{lemma}

\begin{lemma}\label{alpha-split-primes}
For $k,r,M,l\in \mathbb{N}$ with $3 \leq M \leq 2\sqrt{k-1}-1$, we have
\begin{align*}
    \sideset{}{'}\sum\limits_{n\in M_{k,l}}\overline{\alpha}_D(n)^r \leq 2^r\lvert M_{k,l} \rvert  \sum_{m=M+1}^\infty  \Big( \frac{\rho_l}{2}\Big)^{mr} +  2^r \sum_{m=2}^M \Big(\frac{\rho_l}{2}\Big)^{mr} \lvert C_{m,D} \cap  M_{k,l}  \rvert.
\end{align*}
\end{lemma}
\begin{proof}
We use Lemma \ref{newbound} in the proof of Proposition 20 in \cite{Einseleav} and get
\begin{align}\label{equation-split-primes}
    \sideset{}{'}\sum\limits_{n\in M_{k,l}}\overline{\alpha}_D(n)^r \leq \sum_{m=2}^\infty \sum_{n\in M_{k,l}\cap C_{m,D}\setminus C_{m-1,D}}\rho_l^{mr}\alpha_D(n)^r.
\end{align}
With $n\in C_{m,D} \cap C_{m-1,D}$ we have $\alpha_D(n)\leq 2^{-(m-1)}.$ We use this in (\ref{equation-split-primes}), split our sum in two parts, bound  $ \lvert C_{m,D} \cap  M_{k,l}  \rvert \leq \lvert M_{k,l} \rvert$ for the first sum to prove our claim. 
\end{proof}

\begin{definition}
Let $\tilde{M}_{k,l}$ denote the set of $k$-bit integers that are not divisible by the first odd $l$ primes.
\end{definition}

We have the following bound:

\begin{lemma}\label{inclexcl}
For $k\geq 12$, we have
$$2^{k-2.92}\leq \lvert \tilde{M}_{k,2} \lvert \leq 2^{k-2.9}.$$
\end{lemma}
\begin{proof}
Let $A_i$ be the set of $k$-bit integers divisible by $i$. By the inclusion-exclusion principle, we have
\begin{align}\label{inclArg}
     \Big\lvert \hspace{-1.5mm}\bigcup_{ i=2,3,5} \hspace{-1mm} A_i \Big\rvert \notag \hspace{-0.3mm} &= \hspace{-0.3mm}\lvert A_2 \rvert + \lvert A_3 \rvert + \lvert A_5 \rvert \hspace{-0.3mm}- \lvert A_2 \cap \hspace{-0.2mm}A_3 \rvert \hspace{-0.3mm}- \hspace{-0.3mm}\lvert A_2 \cap \hspace{-0.2mm}A_5 \rvert \hspace{-0.3mm}- \hspace{-0.3mm}\lvert A_3 \cap A_3 \rvert + \lvert A_2 \cap \hspace{-0.2mm}A_3 \hspace{-0.2mm}\cap A_5 \rvert \notag \\
    & =\Big\lfloor \frac{2^{k-1}}{2} \Big\rfloor + \Big\lfloor \frac{2^{k-1}}{3} \Big\rfloor + \Big\lfloor \frac{2^{k-1}}{5} \Big\rfloor  - \Big\lfloor \frac{2^{k-1}}{6} \Big\rfloor - \Big\lfloor \frac{2^{k-1}}{10} \Big\rfloor -\Big\lfloor \frac{2^{k-1}}{15} \Big\rfloor +\Big\lfloor \frac{2^{k-1}}{30} \Big\rfloor \notag \\
    & \geq \frac{2^{k-1}}{2}+\frac{2^{k-1}}{3}-1+\frac{2^{k-1}}{5}-1-\frac{2^{k-1}}{6}-\frac{2^{k-1}}{10}-\frac{2^{k-1}}{15}+\frac{2^{k-1}}{30}-1 \notag \\
    &=\frac{11}{15}2^{k-1}-3.
\end{align}
With the same argument, we have that 
\begin{align}\label{inclArg2}
       \Big\lvert \hspace{-1.5mm}\bigcup_{ i=2,3,5} \hspace{-1mm} A_i \Big\rvert \notag \hspace{-0.3mm} &\leq \frac{2^{k-1}}{2}+\frac{2^{k-1}}{3} + \frac{2^{k-1}}{5} - \frac{2^{k-1}}{6} -\frac{2^{k-1}}{10}  -\frac{2^{k-1}}{15}  + \frac{2^{k-1}}{30} +3  \notag\\
    &=\frac{11}{15}2^{k-1}+3 
\end{align}
By inequality (\ref{inclArg}), we have $$\lvert \tilde{M}_{k,2} \rvert = 2^{k-1} - \Big\lvert \bigcup\limits_{i=2,3,5} A_i \Big\rvert \leq \frac{4}{15} 2^{k-1}+3\leq 2^{k-2.9},$$ for $k\geq 12.$
Moreover, we have by inequality (\ref{inclArg2}) that $$\lvert \tilde{M}_{k,2} \rvert = 2^{k-1} - \Big\lvert \bigcup\limits_{i=2,3,5} A_i \Big\rvert \geq \frac{4}{15}2^{k-1}-3 \geq 2^{k-2.92}.$$ for $k\geq 12.$
\end{proof}

The next lemma is easily established and gives us our first upper bound $N_1$, as described in inequality (\ref{NrP}). 
\begin{lemma}\label{alphaboundii}
For $k,M,l\in \mathbb{N}$ with $3 \leq M \leq 2\sqrt{k-1}-1$, we have
\begin{align*}
\sideset{}{'}\sum\limits_{n\in M_{k,l}}\overline{\alpha}_D(n) \leq 2^{k-1.9-M} \frac{{\rho_l}^{M+1}}{2-\rho_l} + 2^{k-2\sqrt{k-1}+1}\rho_l^{M}M(M-1).
\end{align*}
\end{lemma}
\begin{proof}
We use the bound from Lemma \ref{alpha-split-primes}, and set $r=1$. From the proof of Theorem 13 in \cite{Einseleav}, we have that $\lvert C_{m,D} \cap M_{k}\rvert \leq \sum_{j=2}^m 2^{k+1+m-j-\frac{k-1}{j}}$. Moreover, with $M_{k,l} \subseteq M_k$, we have $\lvert  C_{m,D} \cap M_{k,l} \rvert  \leq \lvert C_{m,D} \cap M_{k}\rvert$.

We now follow the proof of Proposition 26 in \cite{Einseleav}. For the first part of the sum, we obtain $ \sum_{m=M+1}^\infty  \Big( \frac{\rho_l}{2}\Big)^{m} = \frac{2^{-M}\rho_l^{M+1}}{2-\rho_l}$ and bound $\lvert M_{k,l} \rvert \leq \lvert M_{k,2}\rvert$ using Lemma \ref{inclexcl}. For the second part of the sum, we use the same argument as in Proposition 26 in \cite{Einseleav}, concluding the proof.
\end{proof}
To get a good estimate for $N_1$, every $M$ that satisfies  $3 \leq M \leq 2\sqrt{k-1}-1$ is a free parameter and will for each $k$ be chosen such that it tightens the bound.

However, for values $k<60$, we will see that we need a tighter bound, and need to do more analysis. The next result serves as an estimate of the number of odd integers, which are both in $M_{k,l}$ and $C_{m,D}.$

\begin{proposition}\label{C_mD1}
    If $m, k$ are positive integers with $m+1 \leq 2\sqrt{k-1}$, then
    $$
    \lvert C_{m,D} \cap M_{k,l} \rvert \leq 2^k \sum_{j=2}^m \frac{2^{m+1-j}-1}{2^{\frac{k-1}{j}}-1}.
$$
\end{proposition}
\begin{proof}
From the proof of Theorem 13 in \cite{Einseleav}, we have that $\lvert C_{m,D} \cap M_k \rvert \leq 2^k \sum_{j=2}^m \frac{2^{m+1-j}-1}{2^{\frac{k-1}{j}}-1}$. The proposition follows with $M_{k,l}\subseteq M_k.$ 
\end{proof}

The next result will serve as our second bound $N_1$ used in inequality (\ref{NrP}).

\begin{lemma}\label{alphaone}
For $k,M,l\in \mathbb{N}$ with $3 \leq M \leq 2\sqrt{k-1}-1$, we have
\begin{equation*}
\sideset{}{'}\sum\limits_{n\in M_{k,l}}\overline{\alpha}_D(n) \leq 2^{1-M} \frac{{\rho_l}^{M+1}}{2-\rho_l}\lvert M_{k,l} \rvert +  2^{k} \sum_{m=2}^M \sum_{j=2}^{ m }    \Big(\frac{\rho_l}{2}\Big)^{m} \frac{2^{m+1-j}-1}{ 2^{\frac{k-1}{j}}-1}.
\end{equation*}
\end{lemma}

\begin{proof}
We use Lemma \ref{alpha-split-primes} with $r=1$. We have  $\sum_{m=M+1}^\infty  \Big(\frac{\rho_l}{2}\Big)^{m} = \frac{2^{-M}{\rho_l}^{M+1}}{2-\rho_l}$ and use Proposition \ref{C_mD1} to get the desired result.
\end{proof}

Similarly to Lemma \ref{alphaboundii}, we choose for each $k$ an $M$ that it minimizes our bound.

For small values of $k$, we can compute the quantity $\lvert M_{k,l} \rvert$ exactly. However, for large values of $k$, the precise computation is not feasible. Hence, it will be useful to upper bound this quantity. We use the trivial fact that $\lvert M_{k,l} \rvert  \leq \lvert M_{k,2} \rvert$.
In the next section, we shall use these preliminary results for bounding $q_{k,t}.$
\section{Intermediate result}
We now establish the main theorem for all $k\geq 42$. However, for smaller values of $k$, further analysis is required, and this will be discussed in Sections \ref{main} and \ref{exact}.

When dealing with values of $k\geq 101$, we use the following theorem to establish our claim in a straightforward manner.

\begin{theorem}[Einsele, Paterson \cite{Einseleav}]\label{probability-estimate}
For $k,l \in \mathbb{N}$ with $k\geq 2$,  we have
\begin{equation*}
  q_{k,1}< k^24^{1.8-\sqrt{k}}  \rho_l^{2\sqrt{k-1}-2}.
\end{equation*}
\end{theorem}

Now, let us introduce an intermediate result that holds for all $t \geq 1$ and $k \geq 42$.

\begin{theorem}\label{frac}
For all $t\geq 1$ and $k\geq 42$ we have
$$
q_{k,t}\leq \Big(\frac{4}{15}\Big)^t.
$$
\end{theorem}

\begin{proof}
By letting $r=1$ in Lemma \ref{4/15-estimate}, it suffices to show that $q_{k,1}\leq 4/19$ for all $k\geq 42$ to prove the theorem. Theorem \ref{probability-estimate} immediately gives us that $q_{k,1} \leq \frac{4}{19}$ for each $k \geq 101,$ since $k^24^{1.8-\sqrt{k}}  \rho_l^{2\sqrt{k-1}-2}$ is a strictly decreasing function for $k\geq 10$ and $q_{101,1}\leq \frac{4}{19}.$ Thus, we get $q_{k,t}\leq \Big(\frac{4}{15}\Big)^t$ for $k\geq 101.$ By inequality (\ref{NrP}), we have
$q_{k,1} \leq \frac{N_1}{N_1+P}$, where $N_1$ is our upper bound for $\sum_{n\in M_{k,l}}^{'}\overline{\alpha}_D(n)$ and $P$ our lower bound for $\pi(2^k)-\pi(2^{k-1})$.  Proposition \ref{k-bit-prime-approx} serves as our value for $P$, and Lemma \ref{alphaboundii} as our first value for $N_1$. For each $k$, we take  $3 \leq M \leq 2\sqrt{k-1}-1$ to be the positive integer that minimizes our upper bound for $q_{k,1}$. Table \ref{tablefor1} displays the results of our computations, proving that $q_{k,1}<4/19 < 0.210527$ for $60 \leq k\leq100$.

\begin{table}
\centering
\begin{minipage}{.25\linewidth}
\centering
\begin{tabular*}{\linewidth}{@{\extracolsep{\fill}}  L | L  L }
k & M\textsubscript{opt} & u_k \\ \hline
    60 & 9 & 0.204541 \\ 
    61 & 9 & 0.196467 \\ 
    62 & 9 & 0.188917 \\
    63 & 9 & 0.181868 \\ 
	64 & 9 & 0.175296 \\ 
	65 & 9 & 0.169176 \\ 
	66 & 9 & 0.163486 \\ 
	67 & 9 & 0.158204 \\ 
	68 & 10 & 0.151309 \\  
	69 & 10 & 0.144807 \\ 
	70 & 10 & 0.138718 \\
	71 & 10 & 0.133020 \\
 	72 & 10 & 0.127693 \\
  	74 & 10 & 0.122717 \\
\end{tabular*}
\end{minipage}
\hspace{1cm}
\begin{minipage}{.25\linewidth}
\centering
\begin{tabular*}{\linewidth}{@{\extracolsep{\fill}}  L | L L  }
k & M\textsubscript{opt} & u_k \\ \hline
 75 & 10 & 0.113743  \\ 
 76 & 10 & 0.109708 \\ 
77 & 11 & 0.105817 \\
 78 & 11 & 0.101064 \\ 
79 & 11 & 0.096609 \\ 
 80 & 11 & 0.092435 \\ 
 81 & 11 & 0.088527 \\ 
 82 & 11 & 0.084870  \\ 
83 & 11 & 0.081449 \\  
 84 & 11 & 0.078251  \\ 
85 & 11 & 0.075262   \\
86 & 11 & 0.072471 \\
 87 & 11 & 0.069865 \\
 88 & 12 & 0.066918 
\end{tabular*}
\end{minipage}
\hspace{1cm}
\begin{minipage}{.25\linewidth}
\centering
\begin{tabular*}{\linewidth}{@{\extracolsep{\fill}}  L | L L  }
k & M\textsubscript{opt} & u_k \\ \hline
 89 & 12 & 0.063918 \\ 
 90 & 12 & 0.061105 \\ 
 91 & 12 & 0.058467 \\
92 & 12 & 0.055994 \\ 
 93 & 12 & 0.053676\\ 
94 & 12 & 0.0515047\\ 
 95 & 12 & 0.0494708\\ 
96 & 12 & 0.0475661 \\ 
 97 & 12 & 0.0457829\\  
98 & 12 & 0.044114 \\ 
 99 & 13 & 0.043620  \\
100 & 13 & 0.040361\\
 & &\\
  & &\\
\end{tabular*}
\end{minipage}
\vspace{2mm}
\caption{All $k<101$ such that $u_k<4/19$ using $M\textsubscript{opt}$ , where $u_k$ is the upper bound for $q_{k,1}$ and $M\textsubscript{opt}$ is the optimal value for $M$ that minimizes $u_k$, where we used the bound for $N_1$ given in Lemma \ref{alphaboundii} and $P=(0.71867)2^k/k$.}
\label{tablefor1}
\end{table}

For the remaining values of $k$, we use Lemma \ref{alphaone} for determining $N_1$. We bound $\lvert M_{k,l} \rvert$ using Lemma \ref{inclexcl}, and for each $k$ we choose $3 \leq M \leq 2\sqrt{k-1}-1$ to be the positive integer that minimizes our upper bound given by inequality (\ref{NrP}) for $q_{k,1}$. Table \ref{table90} shows the results of the computations, proving that $q_{k,1}\leq 4/19$ for $k\geq 42$, and thus concluding our proof.

\begin{table}[!htbp]
\centering
\begin{minipage}{.25\linewidth}
\centering
\begin{tabular*}{\linewidth}{@{\extracolsep{\fill}}  L | L  L }
	k & M\textsubscript{opt} & u_k \\ \hline
	42 & 8 & 0.199683 \\ 
	43 & 8 & 0.189917 \\ 
	44 & 8 & 0.181164  \\ 
	45 & 8 & 0.173352\\ 
	46 & 8 & 0.166410\\ 
    47 & 8 & 0.160268\\
\end{tabular*}
\end{minipage}
\hspace{1cm}
\begin{minipage}{.25\linewidth}
\centering
\begin{tabular*}{\linewidth}{@{\extracolsep{\fill}}  L | L L  }
k & M\textsubscript{opt} & u_k \\ \hline
    48 & 8 & 0.154860 \\
    49 & 9 & 0.147791  \\
    50 & 9 & 0.140038 \\

	51 & 9 & 0.133018 \\ 
	52 & 9 & 0.126677\\ 
	53 & 9 & 0.120964  \\

\end{tabular*}
\end{minipage}
\hspace{1cm}
\begin{minipage}{.25\linewidth}
\centering
\begin{tabular*}{\linewidth}{@{\extracolsep{\fill}}  L | L L  }
k & M\textsubscript{opt} & u_k \\ \hline
    54 & 9 & 0.115831 \\

	55 & 9 & 0.111229 \\ 
	56 & 9 & 0.107117  \\ 
    57 & 10 & 0.102671\\
    58 & 10 & 0.097171 \\
     59 & 10 & 0.092159
\end{tabular*}
\end{minipage}
\vspace{1.2mm}
\caption{All $k<60$ such that $u_k<4/19$ using $M\textsubscript{opt}$ , where $u_k$ is the upper bound for $q_{k,1}$ and $M\textsubscript{opt}$ is the optimal value for $M$ that minimizes $u_k$, where we used the bound for $N_1$ given by Lemma \ref{alphaone} and for $P=(0.71867)2^k/k$.}
 
\label{table90}
\vspace{-3mm}
\end{table}
\end{proof}
\newpage

\section{Main result}\label{main}

In this section, we show the main theorem for $k\geq 18$ and $t\geq 1$. In Section \ref{exact}, we compute $q_{k,t}$ exactly using an equation proven by Arnault, proving that the theorem holds for $k\geq 2$ altogether.

\begin{theorem}\label{frac2}
For all $t\geq 1$ and $k\geq 18$ we have
$$
q_{k,t}\leq \Big(\frac{4}{15}\Big)^t.
$$
\end{theorem}

Theorem \ref{frac} only proves our claim for $k\geq42.$ As the last two terms in the second sum of Lemma \ref{alphaone} dominate our estimate, we divide the set $M_{k,l}$ into two disjoint sets, where the sum terminates earlier for one set, enabling us to lower the value of $k$. Let $M_{k,l,d_1} \subseteq M_{k,l}$ denote the subset consisting of square-free integers having a prime $p\mid n$ such that $\frac{(p-\epsilon_D(p),n-\epsilon_D(n))}{p-\epsilon_D(p)}\geq \frac{1}{3}$. Let $M_{k,l,d_2} \subseteq M_{k,l}$ denote the subset consisting of all integers for which for every prime $p\mid n$ we have $\frac{(p-\epsilon_D(p),n-\epsilon_D(n))}{p-\epsilon_D(p)}< \frac{1}{3}$, unified with the set of non-square free integers.

Let $X_{d_1}$ denote the event that a number chosen uniformly at random from $M_{k,l}$ is composite and lies in $M_{k,l,d_1}$, and $X_{d_2}$ that a number chosen uniformly at random from $M_{k,l}$ is composite and lies in $M_{k,l,d_2}$. Let $N_{r,d_1}$ be the upper bound for $\sum_{n\in M_{k,l,d_1}}^{'}\overline{\alpha}_D(n)^r$ and $N_{r,d_2}$ for $\sum_{n\in M_{k,l,d_2}}^{'}\overline{\alpha}_D(n)^r$ respectively. Then,

\begin{align*}
    q_{k,r} &= \frac{\mathbb{P}[X \cap Y_r]}{\mathbb{P}[Y_r]}=\frac{\mathbb{P}[(X_{d_1} \cup X_{d_2}) \cap Y_r]}{\mathbb{P}[Y_r]}
= \frac{\mathbb{P}[X_{d_1} \cap Y_r]+\mathbb{P}[X_{d_2} \cap Y_r]}{\mathbb{P}[Y_r]} \\
    &= \frac{\sideset{}{'}\sum\limits_{n\in M_{k,l,d_1}}\overline{\alpha}_D(n)^r+\sideset{}{'}\sum\limits_{n\in M_{k,l,d_2}}\overline{\alpha}_D(n)^r}{\sum_{n\in M_{k,l}}\overline{\alpha}_D(n)^r} \\
    &= \frac{\sideset{}{'}\sum\limits_{n\in M_{k,l,d_1}}\overline{\alpha}_D(n)^r+\sideset{}{'}\sum\limits_{n\in M_{k,l,d_2}}\overline{\alpha}_D(n)^r}{\sideset{}{'}\sum\limits_{n\in M_{k,l,d_1}}\overline{\alpha}_D(n)^r+\sideset{}{'}\sum\limits_{n\in M_{k,l,d_2}}\overline{\alpha}_D(n)^r + \pi(2^k)-\pi(2^{k-1})} \\
   & = \frac{N_{r,d_1}+N_{r,d_1}}{N_{r,d_1}+N_{r,d_2}+\pi(2^k)-\pi(2^{k-1})}.
\end{align*}

For ease of notation, let us define the following quantities:
\begin{definition}
Let $\omega(n)$ denote the number of distinct prime factors of $n$ and let $\Omega(n)$ denote the number of prime factors of $n$ counted with multiplicity. Thus, $\omega(n)=s$ and $\Omega(n)=\sum_{i=1}^s r_i$.    
\end{definition}

We use the following lemmas for proving the Theorems \ref{Frac_Cm_Estimate_car} and \ref{Frac_Cm_Estimate_ncar}:

\begin{lemma}\label{lemmaalphaomega}
Let $n=p_1^{r_1}\dots p_s^{r_s}>1$ be odd. Then
$$
\alpha_D(n)\leq 2^{1-\Omega(n)}\prod_{i=1}^s \Big(\frac{2}{p_i}\Big)^{r_i-1}\frac{\gcd(p_i-\epsilon_D(p_i),n-\epsilon_D(n))}{p_i-\epsilon_D(p_i)}.
$$
\end{lemma}
\begin{proof}
This follows from Lemma 10 in \cite{Einseleav}, where we use that $2^{1-s}=2^{1-\Omega(n)}\prod_{i=1}^s 2^{r_i-1}.$
\end{proof}

\begin{lemma}\label{alpha_D_2}
Let $n=p_1^{r_1}\dots p_s^{r_s} \in M_{k,l,d_2}$, meaning that either $\frac{\gcd(p_i-\epsilon_D(p_i),n-\epsilon_D(n))}{p_i-\epsilon_D(p_i)}<\frac{1}{3}$ for all $i$, or $r_i\geq 2$ for some $i$. Then,
$$\alpha_D(n) \leq 2^{-\Omega(n)-1}.$$ 
\end{lemma}

\begin{proof}
Let $n \in M_{k,l,d_2}$. Let us first look at the case where $n$ is square-free. Since $n\in M_{k,l,d_2}$, we must have that $\frac{\gcd(p_i-\epsilon_D(p_i),n-\epsilon_D(n))}{p_i-\epsilon_D(p_i)}< \frac{1}{3}$, and hence $\frac{\gcd(p_i-\epsilon_D(p_i),n-\epsilon_D(n))}{p_i-\epsilon_D(p_i)}\leq \frac{1}{4}$.  Lemma \ref{lemmaalphaomega} directly yields $\alpha_D(n)<\frac{2^{1-\Omega(n)}}{4}=2^{-\Omega(n)-1}$. Now let  us look at the case where $n$ is not square-free, meaning that $r_i\geq 2$ for some $i$. Since $p_1 \geq \tilde{p}_{l+1}$, Lemma \ref{lemmaalphaomega} directly yields that $\alpha_D(n) \leq \frac{2^{1-\Omega(n)}}{\tilde{p}_{l+1}}\leq 2^{-\Omega(n)-1}.$
\end{proof}

The following two lemmas will be used in proving the Theorems \ref{Frac_Cm_Estimate_car} and \ref{Frac_Cm_Estimate_ncar}:

\begin{lemma}[Einsele, Paterson \cite{Einseleav}]\label{part-frac}
If $t\in \mathbb{R}$ with $t \geq 1$, then 
$$ \sum_{n=\lfloor t \rfloor +1 }^\infty \frac{1}{n(n-1)}=\frac{1}{ \lfloor t \rfloor} <\frac{2}{t}.$$
\end{lemma}

\begin{lemma}[Damgård et al. \cite{DamEtAl}]\label{boundn^2}
If $t$ is a real number with $t\geq 1$, then
$$
\sum_{n=\lfloor t \rfloor +1}^\infty \frac{1}{n^2}< \frac{\pi^2-6}{3t}.
$$
\end{lemma}

By treating the two disjoint sets $M_{k,l,d_1}$ and $M_{k,l,d_2}$ differently in our analysis, we get a tighter estimates than Proposition \ref{C_mD1}.
\begin{theorem}\label{Frac_Cm_Estimate_car}
    If $m,k$ are positive integers with $m+1 \leq 2 \sqrt{k-1}$, then
    $$
    \lvert C_{m,D} \cap M_{k,l,d_1} \rvert \leq 2^{k}\sum_{j=2}^m \frac{3}{\prod_{i=1}^{j-1}\tilde{p}_{l+i}}\frac{1}{2^{\frac{k-1}{j}}+1}.$$
\end{theorem}
\begin{proof}
For $n \in C_{m,D} \cap M_{k,l}$, we have by Lemma \ref{lemmaalphaomega} that $2^m>1/\alpha_D(n)>2^{\Omega(n)-1}$, and thus $m+1> \Omega(n)$.  $\Omega(n)\in \mathbb{N}$ implies $\Omega(n)\leq m.$ 
Let $N_D(m,k,j)=\{n \in C_{m,D}\cap M_{k,l,d_1}\mid \Omega(n)=j \}.$ Hence,
\begin{align}\label{CMdMklCar}
\lvert C_{m,D} \cap M_{k,l,d_1}\rvert= \sum_{j=2}^m \lvert N_D(m,k,j) \rvert.
\end{align}
Let $n \in N_D(m,k,j)$ with $2 \leq j \leq m$, and let $p$ be the largest prime factor of $n$. Now, $2^{k-1} < n \leq p^{j}$ implies that $p>2^{(k-1)/j}$. Given $p$ and $d$, where $p$ is a prime with the property that $p>2^{(k-1)/j}$ and $d$ is such that $d \mid p-\epsilon_D(p)$, we want to get an upper bound for the number of $n\in N_D(m,k,j)$ with the largest prime factor $p$ such that $d_D(p,n)=d$. Let $S_{D,k,d, p}=\{n \in M_{k,l, d_1} : p \mid n, d= \frac{p-\epsilon_D(p)} {(p-\epsilon_D(p), n-\epsilon_D(n))}, n \textnormal{ composite}\}$ for $d=1,2,3$. Since $n\in M_{k,l, d_1}$, we know that $n$ is a product of distinct primes $n=p_1\dots p_{j-1}p$, with $p_i \geq \tilde{p}_{l+i}$ for all $i=1, \dots, j-1.$ Thus, $p=\frac{n}{p_1\dots p_{j-1}} \leq \frac{n}{\prod_{i=1}^{j-1} \tilde{p}_{l+i}} < \frac{2^k}{\prod_{i=1}^{j-1} \tilde{p}_{l+i}}.$ The size of $S_{D,k,d,p}$ is at most the number of solutions of the system 
$$
n \equiv 0 \mod p, n \equiv \epsilon_D(n) \mod \frac{p-\epsilon_D(p)}{d}, p < \frac{2^k}{\prod_{i=1}^{j-1} \tilde{p}_{l+i}},
$$
where $d=1,2,3.$
Via the Chinese Remainder Theorem this, set has fewer than $\frac{2^kd}{\prod_{i=1}^{j-1} \tilde{p}_{l+i}}\frac{1}{p(p-\epsilon_D(p))}$ elements. With $p$ and $n$ odd, $(p-\epsilon_D(p))/d(p-\epsilon_D(p), n-\epsilon_D(n))$ must be even. So,
\begin{align*}
\lvert N_D(m,k,j) \rvert &\leq  \sum\limits_{d=1,2,3}
\sum\limits_{\substack{p>2^{(k-1)/j} \\ \\ p-\epsilon_D(p)\in 2\mathbb{Z}}} \frac{2^kd}{\prod_{i=1}^{j-1} \tilde{p}_{l+i}}\frac{1}{p(p-\epsilon_D(p))} \\
&= \frac{2^k}{\prod_{i=1}^{j-1} \tilde{p}_{l+i}} \sum\limits_{d=1,2,3} \sum_{2u > 2^{(k-1)/j}-\epsilon_D(p)} \frac{d}{(2u+\epsilon_D(p))2u}.
\end{align*}
Now let us first look at the case $\epsilon_D(p)=1$. By Lemma \ref{boundn^2}, we have
\begin{align*}
\frac{1}{4d}\sum_{u > \frac{2^{\frac{k-1}{j}}-\epsilon_D(p)}{2d}} \frac{1}{(u+\frac{\epsilon_D(p)}{2d})u} &\leq \frac{1}{4d}\sum_{u > \frac{2^{\frac{k-1}{j}}-1}{2d}}\frac{1}{u^2} \leq \frac{1}{4d}\frac{\pi^2-6}{3 \frac{2^{\frac{k-1}{j}}-1}{2d}}=\frac{\pi^2-6}{6}\frac{1}{2^{\frac{k-1}{j}}-1}.
\end{align*}
Now let us look at the case $\epsilon_D(p)=-1$. By Lemma \ref{part-frac} we have
\begin{align*}
\frac{1}{4d}\sum_{u > \frac{2^{\frac{k-1}{j}}-\epsilon_D(p)}{2d}} \frac{1}{(u+\frac{\epsilon_D(p)}{2d})u} &= \frac{1}{4d}\sum_{u > \frac{2^{\frac{k-1}{j}}+1}{2d}} \frac{1}{(u-1)u} \leq \frac{1}{4d}\frac{2}{\frac{2^{\frac{k-1}{j}}+1}{2d}}= \frac{1}{2^{\frac{k-1}{j}}+1}.
\end{align*}
For $k, j\in \mathbb{N}$ with $j\leq \frac{k-1}{\log_2(-\frac{\pi^2}{\pi^2-12})}$, we have that $\frac{\pi^2-6}{6}\frac{1}{2^{\frac{k-1}{j}}-1} \leq \frac{1}{2^{\frac{k-1}{j}}+1}$. With $j \leq m-2$ and $m+1 \leq 2 \sqrt{k-1}$, this is naturally satisfied.
Hence, we get
\begin{align*}
    \lvert N_D(m,k,j) \rvert \leq  \frac{2^k}{\prod_{i=1}^{j-1}} \sum\limits_{d=1,2,3}\frac{1}{2^{\frac{k-1}{j}}+1}\leq \frac{2^k}{\prod_{i=1}^{j-1}} \frac{ 3 }{2^{\frac{k-1}{j}}+1}.
\end{align*}
 Using this estimate in (\ref{CMdMklCar}) proves the theorem.
\end{proof}

\begin{theorem}\label{Frac_Cm_Estimate_ncar}
If $m,k$ are positive integers with $m+1 \leq 2 \sqrt{k-1}$, then
    $$ \lvert C_{m,D} \cap M_{k,l,d_2} \rvert \leq 2^{k} \sum_{j=2}^{ m-2} \frac{2^{m+1-j} -4 }{2^{\frac{k-1}{j}}+1}.
    \vspace{-1mm}
    $$
\end{theorem}
\begin{proof}
By the same argument as in the proof of Theorem \ref{Frac_Cm_Estimate_car}, we get by applying Lemma \ref{alpha_D_2} that $\Omega(n)\leq m-2$ for $n \in C_{m,D} \cap M_{k,l,d_2}$. Now  let $N_D(m,k,j)=\{n \in C_{m,D} \cap M_{k,l} \mid \Omega(n)=j\}$. We see that 
\begin{align}\label{CmMkl}
\vspace{-1mm}
\lvert C_{m,D} \cap M_{k,l} \rvert = \sum_{j=2}^{m-2} \lvert N_D(m,k,j) \rvert.
\vspace{-1mm}
\end{align}
Let $n\in N_D(m,k,j)$ with $2 \leq j \leq m-2$, and let $p$ be the largest prime factor of $n$. Now $2^{k-1}< n\leq p^j$ implies that $p>2^{(k-1)/j}$. Let $d_D(p,n)=(p-\epsilon_D(n))/(p-\epsilon_D(n), n-\epsilon_D(n))$. Lemma \ref{lemmaalphaomega} implies that $2^m>1/\alpha_D(n)\geq 2^{\Omega(n)-1}d_D(p,n)=2^{j-1}d_D(p,n)$, so we must have $d_D(p,n)<2^{m+1-j}.$ That $n\in M_{k,l,d_2}$ implies  $d_D(p,n)> 3.$ Given $p,d$, where $p$ is a prime with the property that $p>2^{(k-1)/j}$ and $d$ is such that  $d \mid p-\epsilon_D(p)$ and $d<2^{m+1-j}$, we want to get an upper bound for the number of $n\in N_D(m,k,j)$ with largest prime factor $p$ such that $d_D(p,n)=d$. Let  $S_{D,k,d,p}=\{n\in M_{k,l} : p \mid n, d=\frac{p-\epsilon_D(p)}{(p-\epsilon_D(p), n-\epsilon_D(n))}, n \text{ composite}\}.$
The size of the set $S_{D,k,d,p}$ is at most the number of solutions of the system 
\begin{eqnarray*}
    n \equiv 0 \bmod p, &n\equiv \epsilon_D(n) \bmod \frac{p-\epsilon_D(p)}{d}, & p<n<2^k.
\end{eqnarray*}
\vspace{-1mm}
Via the Chinese Remainder Theorem, this is less than $\frac{2^{k}d}{p(p-\epsilon_D(p))}$. 

If $S_{D,k,d,p} \neq \emptyset$, then there exists an $n \in S_{D,k,d,p}$ with ${(n-\epsilon_D(n),p-\epsilon_D(p))=(p-\epsilon_D(p))/d}$. Again $(p-\epsilon_D(p))/d=(p-\epsilon_D(p), n-\epsilon_D(n))$ must be even.
Hence,
\begin{align*}
    \lvert N_D(m,k,j) \rvert & \hspace{-0.5mm} \leq \hspace{-2.2mm} \sum_{p>2^{(k-1)/j}} \hspace{-2mm} \sum\limits_{\substack{d \mid p-\epsilon_D(p) \\ 3 < d<2^{m+1-j} \\ (p-\epsilon_D(p))/d \in 2\mathbb{Z}} } \hspace{-2mm}\frac{2^k d}{p(p-\epsilon_D(p))} \hspace{-0.5mm} =2^k \hspace{-2mm}\sum_{3 < d<2^{m+1-j}} \hspace{-2.2mm} \sum\limits_{\substack{p>2^{(k-1)/j} \\ d \mid p-\epsilon_D(p) \\ (p-\epsilon_D(p))/d \in 2\mathbb{Z}}}\hspace{-2.4mm}  \frac{d}{p(p-\epsilon_D(p))}.
\end{align*}

\vspace{-1mm}
Now, for the inner sum we have, 
\begin{align*}
\sum\limits_{\substack{p>2^{(k-1)/j} \\ d \mid p-\epsilon_D(p) \\ \frac{p-\epsilon_D(p)}{d} \in 2\mathbb{Z}}} \frac{d}{p(p-\epsilon_D(p))} 
<\hspace{-4mm}\sum_{2ud>2^{\frac{k-1}{j}}-\epsilon_D(p)}\hspace{-2mm}\frac{d}{(2ud+\epsilon_D(p))2ud}=\frac{1}{4d}\hspace{-2mm}\sum_{u > \frac{2^{\frac{k-1}{j}}-\epsilon_D(p)}{2d}} \frac{1}{(u+\frac{\epsilon_D(p)}{2d})u}.
\end{align*}

By the same argument as in Theorem \ref{Frac_Cm_Estimate_car} we get
\begin{align*}
    \lvert N_D(m,k,j) \rvert \leq  2^k \sum\limits_{3 < d<2^{m+1-j}}\frac{1}{2^{\frac{k-1}{j}}+1}\leq 2^{k}  \frac{ 2^{m+1-j} -4 }{2^{\frac{k-1}{j}}+1}.
\end{align*}
Using this estimate in (\ref{CmMkl}) concludes the proof.
\end{proof}

\begin{theorem}\label{alphafort}
Let $m,k, l, r\in \mathbb{N}$, $k\geq 2$ with $m+1 \leq 2\sqrt{k-1}$. Then, 
\vspace{-1.5mm} 
\begin{align*}
\sideset{}{'}\sum\limits_{n\in M_{k,l}}\overline{\alpha}_D(n)^r \leq &2^{r(1-M)} \lvert M_{k,l} \rvert \frac{\rho_l^{(M+1)t}}{2^t-\rho^r}+2^{k+r}\Bigg( \sum_{m=2}^M \sum_{j=2}^{m} \Big( \frac{\rho_l}{2} \Big)^{mr}  \frac{3}{\prod_{i=1}^{j-1}\tilde{p}_{l+i}}\frac{1}{2^{\frac{k-1}{j}}+1} \\
&+\sum_{m=2}^M \sum_{j=2}^{m-2} \Big( \frac{\rho_l}{2} \Big)^{mr} \frac{ 2^{m+1-j}-4}{ 2^{\frac{k-1}{j}}+1}\Bigg).\end{align*}
\end{theorem}

\begin{proof}
With Lemma \ref{newbound} we get
\begin{align*}
\sideset{}{'}\sum\limits_{n\in M_{k,l}}\overline{\alpha}_D(n)^t =& \sum_{m=2}^\infty \sum_{n \in C_{m,D} \cap M_{k,l} \setminus C_{m-1,D}} \overline{\alpha}_D(n)^r \\
\leq& \sum_{m=2}^\infty \sum_{n \in C_{m,D} \cap M_{k,l} \setminus C_{m-1,D}} \rho_l^{mt} 2^{-(m-1)r} \\
\leq& 2^r \Big( \lvert M_{k,l} \rvert \sum_{m=M+1}^\infty  \Big( \frac{\rho_l}{2} \Big)^{mr} +  \Big( \sum_{m=2}^M  \Big( \frac{\rho_l}{2} \Big)^{mr} \mid M_{k,l,d_1} \cap C_{m,D} \mid \\
&+ \sum_{m=2}^M \Big( \frac{\rho_l}{2} \Big)^{mr} \mid M_{k,l,d_2} \cap C_{m,D} \mid \Big)\Big).
\end{align*}
With $\sum_{m=M+1}^\infty \big( \frac{\rho_l}{2} \big)^{mr} = 2^{-Mr} \frac{\rho_l^{(M+1)r}}{2^r-\rho^r}$ and Theorem \ref{Frac_Cm_Estimate_car} and \ref{Frac_Cm_Estimate_ncar} we are done.
\end{proof}

\textit{Proof of Theorem \ref{frac2}.} We use inequality (\ref{NrP}) with Theorem \ref{alphafort} for $N_1$ and $N_2$. For each $k$, we choose $M$ to be the positive integer $3 \leq M \leq 2\sqrt{k-1}-1$ that minimizes each upper bound in inequality (\ref{NrP}) for $q_{k,1}$ and $q_{k,2}$ with $P=(0.71867)\frac{2^k}{k}$. This shows that $q_{k,1}\leq 4/19$ for $k\geq 34$. For $30 \leq k \leq 33$ we have that $q_{k,1}\leq 4/15$ and $q_{k,2} \leq 16/241$. So, $q_{k,t} \leq (4/15)^t$ for $k\geq 30.$

\begin{table}[H]
\centering
\begin{minipage}{0.42\linewidth}
\centering
\begin{tabular*}{\linewidth}{  L | L | L | L | L}
	k & M_{opt,1} & v_{k,1} & M_{opt, 2} & v_{k,2} \\ \hline
    30 & 6 & 0.239294 & 8 & 0.000602\\
    31 & 6 & 0.235818 & 8 & 0.000544\\ 
    32 & 7 & 0.232670 & 8 & 0.000360\\ 
    33 & 7 & 0.220337 & 9 & 0.000314\\ 
    34 & 7 & 0.209791 & & \\
    35 & 7 & 0.200868 & & \\ 
\end{tabular*}
\end{minipage}%
\hspace{1cm}
\begin{minipage}{0.23\linewidth}
\centering
\begin{tabular*}{\linewidth}{  L | L | L}
	k & M_{opt,1} & v_{k,1} \\ \hline
    36 & 7 & 0.193406 \\ 
    37 & 7 & 0.187248 \\
    38 & 7 & 0.182247 \\
    39 & 7 & 0.178267 \\ 
    40 & 7 & 0.175183 \\
    41 & 8 & 0.166822  
\end{tabular*}
\end{minipage}%
\vspace{2mm}
\caption{Let $v_{k,1}$ be the upper bound for $q_{k,1}$ and $v_{k,2}$ for $q_{k,2}$ respectively, where we use Theorem \ref{alphafort} with Lemma \ref{inclexcl} for bounding $\lvert M_{k,l} \rvert$ and  $P=(0.71867)2^k/k.$  $M_{\textnormal{opt},1}$ is the optimal value for $M$ that minimizes $v_{k,1}$ and $M_{\textnormal{opt},2}$ for $v_{k,2}$ respectively. $v_{k,1}\leq 4/19$ ror $ 34 \leq k \leq 41$, and for $30 \leq k \leq 33$ we have $v_{k,1}\leq 4/15$ and $v_{k,2} \leq 16/241$.}
\end{table}

We now compute the exact values of $\lvert M_{k,l} \rvert $ and $\pi(2^k)-\pi(2^{k-1})$ to get improved results in Theorem \ref{alphafort} with $r=1,2$. For each $k$, we choose $3 \leq M \leq 2 \sqrt{k-1}-1$  that minimizes the upper bound in (\ref{NrP}). With this, we have $q_{k,1} \leq 4/19$ for $k=27, 28, 29$, and $q_{k,1} \leq 4/15$ and $q_{k,2} \leq 16/241$ for $17 \leq k \leq 26$, which proves that $q_{k,t} \leq (4/15)^t$ for $17 \leq k\geq 29$, see Table \ref{tablek1729}.

\begin{table}[H]
\centering
\begin{minipage}{0.46\linewidth}
\centering
\begin{tabular*}{\linewidth}{  L | L | L | L | L}
	k & M_{opt,1} & v_{k,1} & M_{opt, 2} & v_{k,2} \\ \hline
    17 & 4 & 0.253449 & 6 & 0.004786 \\
    18 & 4 & 0.256262 & 6 & 0.004075 \\ 
    19 & 4 & 0.260073 & 6 & 0.003510 \\ 
    20 & 5 & 0.247789 & 6 & 0.003088 \\ 
    21 & 5 & 0.235446 & 6 & 0.002760 \\ 
    22 & 5 & 0.226473 & 7 & 0.001935 \\ 
    23 & 5 & 0.220211 & 7 & 0.001650 \\
\end{tabular*}
\end{minipage}%
\hspace{1cm}
\begin{minipage}{0.46\linewidth}
\centering
\begin{tabular*}{\linewidth}{  L | L | L | L | L}
	k & M_{opt,1} & v_{k,1} & M_{opt, 2} & v_{k,2} \\ \hline
    24 & 5 & 0.216189 & 7 & 0.001424 \\
    25 & 5 & 0.214003 & 7 & 0.001246 \\
    26 & 5 & 0.213406 & 8 & 0.000926 \\
    27 & 6 & 0.209426 & & \\
    28 & 6 & 0.197899 & & \\
    29 & 6 & 0.188524 & & \\
    & & & &
\end{tabular*}
\end{minipage}%
\vspace{2mm}
\caption{Let $v_{k,1}$ be the upper bound for $q_{k,1}$ and $v_{k,2}$ for $q_{k,2}$ respectively, using Theorem \ref{alphafort} and the exact values for $\lvert M_{k,l} \rvert$ for the bounds $N_1$ and $N_2$, and the exact values for $\pi(2^k)-\pi(2^{k-1})$. We let $M$ to minimize $v_{k,1}$ and $v_{k,2}$, denoted as $M_{\textnormal{opt},1}$ and $M_{\textnormal{opt},2}$ respectively.
This gives us for  $k= 27, 28, 29$ that $v_{k,1}\leq 4/19$, $v_{k,1}\leq 4/15$ and $v_{k,2} \leq 16/241$ for $17 \leq k \leq 26.$ }
\label{tablek1729}
\end{table}
 
\vspace{-8mm}

\section{Exact values} \label{exact}
In this section, we finally prove the main theorem.
\begin{theorem}
For $k\geq 2, t\geq 1$ we have $q_{k,t} \leq \Big( \frac{4}{15} \Big)^t.$
\end{theorem}

The approach mentioned above does not prove the desired results for $k\leq 17.$ We now suppose divisibility by the first two odd primes instead of the first nine, which in our case is a stronger assumption for $q_{k,t}$. We have an exact formula for $SL(D,n)$ given in (\ref{SLDn}), which we will use. We consider all odd $k$-bit integers that are not divisible by 3 and 5 and store them in a list. To compute $SL(D,n)$, all of the prime factors of $n$ must be determined computationally, so this can only be computed for small values of $k$. We store all primes less than the square root of $2^{16}$ in a list. Moreover, for each $n \in M_{k,2}$, we determine $q, k_1, s$, and $q_i$ for all $1 \leq i \leq s$. We then calculate $\overline{\alpha}_D(n)$. Additionally, we only sum over the integers with $\gcd(n,2D)=1$. For the bound to hold for every $D$, we take the one that maximizes $\sideset{}{'}\sum\limits_{n\in M_{k,l}}\overline{\alpha}_D(n).$

\begin{table}[!htbp]
\centering
\label{pk1exactsmallk}
\begin{minipage}{.08\linewidth}
\centering
\begin{tabular*}{\linewidth} {L |L  }
	k & q_{k,1}    \\ \hline
    2 & 0  \\ 
    3 & 0  \\ 
    4 & 0  \\
    5 & 0
 \end{tabular*}
\end{minipage}%
\hspace{1cm}
\begin{minipage}{.14\linewidth}
\centering
\begin{tabular*}{\linewidth} {L |L   }
	k & q_{k,1}    \\ \hline
	6 & 0.009725  \\ 
	7 & 0.027481  \\ 
	8 & 0.019684  \\ 
    9 & 0.016090 
 \end{tabular*}
\end{minipage}%
\hspace{1cm}
\begin{minipage}{.15\linewidth}
\centering
\begin{tabular*}{\linewidth} {L |L  }
	k & q_{k,1}    \\ \hline
    10  & 0.012924 \\
    11  & 0.008977 \\
    12 & 0.006131	\\
    13 & 0.006737
 \end{tabular*}
\end{minipage}%
\hspace{1cm}
\begin{minipage}{.153\linewidth}
\centering
\begin{tabular*}{\linewidth} {L |L  }
	k & q_{k,1}    \\ \hline
   14 & 0.003987 \\ 
	15 & 0.001641 \\ 
	16 & 0.001095 \\ 
& 
 \end{tabular*}
\end{minipage}%
\vspace{2mm}
\caption{The exact values for $q_{k,1}$ for $2 \leq k \leq 16.$}
\end{table}
\vspace{-5mm}

\hspace{-2mm}Together with Theorems \ref{frac} and \ref{frac2}, this proves that $q_{k,t} \leq \Big( \frac{4}{15} \Big)^t$ for $k\geq 2, t\geq 1.$
\section{Average case behaviour on incremental search}
So far, we have explored a method for generating primes by selecting a fresh and random $k$-bit integer and using Algorithm \ref{algorithm} for primality testing until a passing candidate is found.
However, an often recommended alternative is to choose a random starting point  $n_0 \in M_k$, test it for primality, and if it fails, consecutively test  $n_0+2,  n_0+4 , \dots $ until one is found that passes all stages of the test. Numerous adaptations are possible, such as other step sizes and various sieving techniques, yet the basic principle remains unchanged.
This method, commonly known as \textit{incremental search}, offers several practical advantages. It is more efficient in using random bits and test division by small primes can be conducted much more efficiently compared to the conventional ``uniform choice'' method. A drawback is a bias in the distribution of the generated primes.\\
The key advantage lies in the fact that a complete trial division is only necessary for the starting candidate $n_0.$ The remainders $r_p \equiv n_0 \mod p$ are computed for all primes $p<B$ below a given threshold $B$ and stored in a table. As the candidate sequence progresses, the values in the table are efficiently updated by adding 2 to each stored remainder modulo $p$, so $n_i=n_0+2i\equiv r_p +2i \mod p$, where $n_i$ is the $i$-th candidate and $i\in \mathbb{N}_0$. The candidate passes the trial division stage if none of the table values are equal to 0.

\subsection{The incremental search algorithm}
The analysis of the average-case error probability done in this paper and in \cite{Einseleav} depends on the assumption that the candidates are independent, hence the method cannot be directly taken over to the incremental search method. The average case error behaviour of the incremental search algorithm of the Miller-Rabin test was studied in \cite{incrementalBrandt}. Notably, no analysis of the average-case error behaviour for the strong Lucas test exists. We now give a more precise version of the algorithm.

\begin{algorithm}
\caption{\textsc{PrimeIncLuc}($t, k, s$)}
\textbf{Input:} Bit-size $k\in \mathbb{N}$, testing rounds $t \in \mathbb{N}$, maximum number of candidates before returning ``fail'' $s\in \mathbb{N}$.\\
\textbf{Output:} First probable prime found or ``fail'' after $n_0 + 2s$ iterations.
\begin{enumerate}
    \item Choose an odd $k$-bit integer $n_0$ uniformly at random
    \item $n=n_0$ 
    \item If $n$ is divisible by 2, 3, and 5: set $n=n+2$. If $n\geq n_0+2s$ \textbf{output} "fail" and \textbf{stop}. Else, go to step 3. 
    \item Else execute the following loop until it stops: \\
 For $i=1$ to $t$:
        \begin{itemize}
            \item Perform the strong Lucas test to $n$ with randomly chosen bases. 
            \item If $n$ fails any round of test, set $n=n+2$. If $n\geq n_0+2s$, \textbf{output} "fail" and \textbf{stop}. Else, go to step 3.
            \item Else output $n$ and \textbf{stop}
        \end{itemize}
    \end{enumerate}
\end{algorithm}

To enhance the algorithm's efficiency, we can incorporate test division by additional small primes before applying the strong Lucas test. Regardless of the number of primes used, the optimized algorithm's error probability remains at most that of \textsc{PrimeIncLuc}. This is because test division can never reject a prime, only improving our chances of rejecting composites. However, the following analysis does not explore the error probability of the optimized version.

\subsection{Error estimates of the incremental search algorithm}

In this section, we focus on the probability that \textsc{PrimeincLuc} outputs a composite. Let $y_{k,t,s}$ denote the probability that one execution of the loop (steps 1(a) and (b)) outputs a composite number.

\begin{definition}
Let $$\overline{C}_{m,D}=\{n \in \mathbb{N} : \gcd(n,2D)=1, n \textnormal{ composite and } {\overline{\alpha}_D}(n)=\frac{SL(D,n)}{n-\epsilon_D(n)-1}>2^{-m}\}.$$
  \end{definition}
\vspace{-3mm}
\begin{lemma}\label{boundCmD}
$\overline{C}_{m,D}  \subseteq C_{1.2m,D}.$
\end{lemma}

\begin{proof}
Since $7=\tilde{p}_3$ is the third odd prime, we have by Lemma 2 for every $n\in \tilde{M}_{k,2}$ that
\vspace{-2mm}
\begin{equation}\label{43alpha}
 \overline{\alpha}_D(n) \leq \rho_1^m \alpha_D(n) = \Bigg(\frac{8}{7}\Bigg)^m \alpha_D(n).
\end{equation}
For $n \in \overline{C}_{m,D},$ we have by inequality (\ref{43alpha}) that $2^{-m}<\overline{\alpha}_D(n)\leq \Big(\frac{8}{7}\Big)^m \alpha_D(n).$ Hence, $2^{-m}\cdot \Big(\frac{7}{8}\Big)^m< \alpha_D(n)$ and since $2^{-1.2m}<2^{-m}\Big( \frac{7}{8} \Big)^m$, we get $2^{-1.2m} < \alpha_D(n).$ Thus, for $n\in \overline{C}_{m,D}$, we have $n \in C_{1.2m,D}.$  
\end{proof}
Let us define the set $\mathcal{D}_{m,D} = \{ n \in \tilde{M}_{k,2} \mid [n, \dots, n+2(s-1)) \cap \overline{C}_{m,D} \neq \emptyset \}$ for $m \geq 3.$  A number belonging to $\overline{C}_{m,D}$ can be in at most $s$ distinct intervals of the form $[n, \dots, n+2(s-1))$, Therefore, the next lemma easily follows.

\begin{lemma}\label{lemmaDm}
    $\mathcal{D}_{m,D} \subset \mathcal{D}_{m+1,D}$ and $\lvert \mathcal{D}_m \rvert \leq s \lvert \tilde{M}_{k,2} \cap \overline{C}_{m,D} \rvert.$
\end{lemma}
The motivation behind defining the sets $\mathcal{D}_{m,D}$ is that, if we luckily select a starting point $n_0$ for the inner loop that is \textit{not} a member of $\mathcal{D}_{m,D}$, all composites tested before the loop terminates are going to pass with a probability of at most $2^{-m}.$ This is reflected in the bound $y_{k,s,t}$ as outlined below:

\begin{theorem}\label{thrmincrmprob}
    Let $s=c \ln(2^k)$ for some constant $c.$ Then for any $3\leq M \leq 2\sqrt{k-1}-1$, we have
    \vspace{-3mm}
    $$
    y_{k,t,s} \leq 0.5(ck)^2 \sum_{m=3}^{\lceil 1.2M \rceil } \frac{ \lvert C_{m,D} \cap M_{k} \rvert }{\lvert \tilde{M}_{k,2} \rvert }2^{-t(m-1)}+0.7 ck 2^{-tM}.$$
\end{theorem}

\begin{proof}
    Let $E'$ denote the event of outputting a composite within the inner loop, and let $D_{m,D}$ be associated with the event that the starting point $n_0$ is in $\mathcal{D}_{m,D}$.
    Let $X^\mathsf{c}$  be defined as the complement of the event $X$.
    Then, 
    \vspace{-2mm}
\begin{align*}
    y_{k,t,s} &= \sum_{m=3}^M \mathbb{P}[E'\cap (D_{m,D} \setminus D_{m-1,D})]+\mathbb{P}[E' \cap D_{m,D}^\mathsf{c}] \\
    &\leq \sum_{m=3}^M \mathbb{P}[D_{m,D}]\mathbb{P}[E' \mid (D_{m,D} \setminus D_{m-1,D})]+ \mathbb{P}[E' \cap D_{m,D}^\mathsf{c}].
\end{align*}
Consider the scenario where a fixed $n_0 \not\in \mathcal{D}_{m,D}$ is chosen as the starting point. In this case, no candidate $n$ that we test will belong to $\overline{C}_{m,D} \cap \tilde{M}_{k,2} $, and thus, each candidate will pass all tests with probability at most $2^{-mt}$. The probability of outputting a composite is maximized when all numbers in the considered interval are composite. Consequently, in this scenario, we accept one of the candidates with a probability of at most $s 2^{-mt}.$ Combining this observation with Lemma \ref{lemmaDm} and using the fact that $\tilde{M}_{k,2}\leq M_k$, we obtain
\vspace{-2mm}
\begin{align}\label{yktsineq}
    y_{k,t,s} &\leq s^2 \sum_{m=3}^M \frac{ \lvert \overline{C}_{m,D} \cap M_{k} \rvert }{\lvert \tilde{M}_{k,2} \rvert }2^{-t(m-1)}+s 2^{-tM} \notag\\
    &\leq 0.5(ck)^2 \sum_{m=3}^M \frac{ \lvert \overline{C}_{m,D} \cap M_{k} \rvert }{\lvert \tilde{M}_{k,2} \rvert }2^{-t(m-1)}+0.7 ck 2^{-tM}.
\end{align}
With Lemma \ref{boundCmD} and a substitution we get the desired result.
\end{proof}
By the proof of Theorem 13 in \cite{Einseleav}, we have $\lvert C_{m,D} \cap M_k \rvert \leq 2^{k+1} \sum_{j=2}^m 2^{m-j-\frac{k-1}{j}}$. 
Using Theorem \ref{thrmincrmprob} and Lemma \ref{inclexcl} in inequality (\ref{yktsineq}), we get
$$
y_{k,t,s} \leq 2^{3.42+t}(ck)^2 \sum_{m=3}^{\lceil 1.2M \rceil} 2^{m(1-t)} \sum_{j=2}^m 2^{-j-\frac{k-1}{j}}+0.7 ck 2^{-tM}.
$$ We can now directly derive numerical estimates for $y_{k,t,s}$ for any given value of $s$. Table \ref{table-logykts} shows concrete values for $c=1,5$ and $10$, where the value of $M$ was chosen such that it minimizes the estimate.
\begin{table}[H]
\centering
 \begin{tabular}{L|  L |  L L L L L L L L L L}   
c & k\backslash t  & 1& 2 & 3 & 4 & 5 & 6 & 7 & 8 & 9 & 10 \\ \hline
1 &	100  & 0 & 6 & 12 & 17 & 21 & 25 & 28 & 31& 33 & 35 \\ 
&	200  & 3 & 15 & 24 & 32  & 38 &43 & 48 & 52 & 56 & 59 \\ 
&	400  & 11 & 30 & 42 & 53 & 62 & 69 & 76 & 83& 89 & 94\\ 
&	512  & 15 &  36 & 51 & 62 & 72 & 81 & 83 & 97 & 104 & 110 \\
&	1024 & 31 & 61 & 81 & 98 & 112 & 125 & 137 & 148& 158 & 167 \\ 
&	2048 & 54 & 96 & 125 & 149 & 169 & 188 & 205 & 220& 235 & 249\\
&	4096 & 89  & 147 & 187 & 221 & 251 & 277 & 302 & 324& 345 & 365\\ 
\hline
5 &	100  & 0 & 2 & 8 & 12 & 17 & 20 & 23 & 26 &28 &31\\
&	200  &  0& 11 & 20 & 27 & 33 & 38 & 43& 47& 51 &55\\ 
&	400  &  7& 25 & 38 & 48 & 57 & 65 & 72 & 78& 84 & 90\\ 
&	512  &  11& 32 & 46 & 58 & 68 & 77 & 85& 92& 99 & 106\\
&	1024 &  26& 56 & 76 & 93 & 108 & 120 & 132& 143& 153 & 163\\ 
&	2048 &  50& 91 & 120 & 144 & 165 & 183 & 200 & 216& 230 &244\\
&	4096 &  84 & 142 & 183 & 217 & 246 & 273 & 297& 320& 341 &361\\ 
\hline
10 & 100  & 0 & 0 & 6 & 10 & 15& 18 & 21& 24& 26 &29\\ 
&	200  & 0 & 9 & 18 & 25 & 31& 36 & 41 &45 & 49& 53\\ 
&	400  & 5 & 23 & 36 & 46 & 55 & 63 & 70 &76 & 82& 88\\ 
&	512  & 9 & 30  & 44 & 56  & 66 & 75 & 83 & 90& 97& 104\\
&	1024 & 24 & 54 & 74 & 91 & 106 & 118 &  130 &141 & 151& 161\\ 
&	2048 & 48 & 89 & 118 & 142 & 163 & 181 & 198 &214& 228& 242\\
&	4096 & 82 & 140 & 181 & 215 & 244 & 271 & 295& 318 & 339 &359
\end{tabular}
\caption{ \label{table-logykts} Lower bounds of $-\log_2(y_{k,t,s})$ as a function of $k$ and $t$, where $s=c\ln(2^k)$ with $c=1, 5, 10.$}
\end{table}
\vspace{-4mm}
The next proposition gives a rough idea of how bound's behaviour for large $k$.
\begin{proposition}
Given constants $c$, where $s=c  \ln(2^k)$, and $t$, the function $y_{k,t,s}$ with respect to $k\geq 11$ satisfies
    $$
    y_{k,t,s} \leq \lambda k^32^{-\sqrt{k}} $$
    for a constant $\lambda.$
\end{proposition}
\begin{proof}
It is sufficient to show the proposition for $t=1.$ Using Lemma 12 from \cite{Einseleav} and the proof of Theorem 13 in \cite{Einseleav}, we have $\lvert C_{m,D} \cap M_k \rvert \leq 2^{k+1} \sum_{j=2}^m 2^{m-j-\frac{k-1}{j}}$ and Lemma \ref{inclexcl} with $3-2\sqrt{k-1}<-\sqrt{k}$ for $k\geq 11$ we get
\begin{align}\label{yktsbound}
\frac{\lvert C_{m,D} \cap \tilde{M}_{k,2} \rvert }{\lvert \tilde{M}_{k,2} \rvert}  \leq 2^{m+3.92} \sum_{j=2}^{m}2^{-j-\frac{k-1}{j}} \leq m2^{m+3-2\sqrt{k-1}}<m2^{m-\sqrt{k}},
\end{align}
 for $k\geq 18$. Using inequality (\ref{yktsbound}) in Theorem \ref{thrmincrmprob} with  $M+1\leq 2\sqrt{k-1}$, we get
\begin{align*}
y_{k,t,s} &\leq 0.5(ck)^2 2^{-\sqrt{k}}\sum_{m=3}^{\lceil1.2M\rceil} m +0.7ck2^{-M} =(ck)^2 2^{-\sqrt{k}}(M-1)+0.7ck2^{-M}\\
& <2c^2k^{2.5}2^{-\sqrt{k}}+0.7ck2^{-M}\leq \lambda k^32^{-\sqrt{k}},
\end{align*}
for some constant $\lambda \in \mathbb{R}.$
\end{proof}

In our analysis, we have focused on the probability of a single iteration of the loop producing a composite output. To assess the overall error probability of the algorithm, we observe that the inner loop always terminates when initiated with a prime as a starting point. According to Lemma \ref{inclexcl}, this termination occurs with a probability of $(\pi(2^k)-\pi(2^{k-1}))/\lvert \tilde{M}_{k,2}  \rvert \geq 5.3/k$. Consequently, there exists an exponentially small, relative to $k$, probability that the number of iterations exceeds $k^2.$ Furthermore, it is noteworthy that the error probability of our algorithm is not more than that of a procedure running the inner loop up to $k^2$ times, outputting a composite only in the event of all iterations failing. This observation substantiates an upper bound on the expected running time. Let $Y_{k,t,s}$ denote the probability of \textsc{PrimeincLuc} outputting a composite. Thus, we arrive at the inequality: 
$$
Y_{k,t,s}\leq k^2 y_{k,t,s}+\Big(1-\frac{5.3}{k}\Big)^{k^2}.
$$
\vspace{-2.5mm}
\appendix
\section{Appendix}
Let $n=\prod_{i=1}^s p_i^{r_i}$ be the prime decomposition of an integer $n.$
\vspace{-1mm}
\subsection{The Fermat, Miller-Rabin Primality and Lucas Test}\label{fermat}
The Fermat test is a simple primality tests and exploits the following theorem:
\begin{theorem}[Fermat's Little Theorem]
Let $p$ be a prime number. For all $a$ relatively prime to $p$, we have
$$
a^{p-1} \equiv 1 \bmod p.
$$
\end{theorem}
\vspace{-1mm}
To test whether $p$ is prime, we can check if a random integer $a$ coprime to $p$ satisfies Fermat's Little Theorem. This is called the \textit{ Fermat test}. A \textit{pseudoprime base $a$}, or \textit{psp$(a)$}, is a composite number $n$ such that $a^{n-1} \equiv 1 \bmod n$. The next theorem counts the number of bases that make an integer pass the test:
\begin{theorem}[Baillie et al. \cite{Lucas-Baillie}]\label{number-Fermat}
Let $n=\prod_{i=1}^sp_i^{r_i}$ be a positive integer. Then, the number of bases $ a \bmod n$ for which $n$ is a $psp(a)$ is given by
$$
F(n)=\prod_{i=1}^s \gcd(n-1,p_i-1).
$$
\end{theorem}
\vspace{-1mm}
Unfortunately, infinitely many integers satisfying the thereom for all $a$ coprime to $n$ exist and are known as \textit{Carmichael numbers} \cite{alford1994there}. Consequently, Carmichael numbers can never be identified as composites by the test, making the test impractical for standalone implementation. 

Let us look at a more stringent variant of Fermat's Little Theorem.
\begin{theorem}[The Miller-Rabin Theorem]\label{Miller-Rabin-Theorem}
Let $n>1$ be an integer, and write $n-1=2^\kappa q$, for $q,k \in \mathbb{N}$, where $q$ is odd. Then, $n$ is a prime if and only if for every $a \not\equiv 0 \bmod n$
one of the following is satisfied:
\begin{align}\label{Miller}
    &a^q \equiv 1 \bmod n \notag\\
    \text{or}\\
    &\text{there exists an integer $i < \kappa $ with }a^{2^iq} \equiv -1 \bmod n. \nonumber
\end{align}
\end{theorem}
We can extend this to a primality test, known as the \textit{Miller Rabin test}, by testing property (\ref{Miller}) for several bases $a$. If the property holds for some pair $n, a$, we say $n$ is a \textit{strong pseudoprime base $a$}, in short \textit{spsp}$(a)$.

Besides the strong Lucas test, there is also its weaker variant, the so-called  \textit{Lucas test}, which relies on the following theorem:
\begin{theorem}[Baillie et al. \cite{Lucas-Baillie}] \label{thrm:Lucas}
Let $U_p(P,Q)$ be a Lucas sequence of the first kind. If $p$ is an odd prime such that $(p,QD)=1$, then we have that 
\begin{align}\label{lucasseq}
    U_{p-\epsilon_D(p)} \equiv  0 \bmod p.
\end{align}
\end{theorem}

An integer that satisfies congruence (\ref{lucasseq}) using parameters $(P,Q)$ is called a \textit{Lucas pseudoprime} for $(P,Q),$ short \textit{lpsp($P,Q)$}. We obtain the \textit{Lucas test} by repeatedly checking congruence (\ref{lucasseq}) for various pairs ($P,Q$). 

For a fixed integer $D$, the number of parameter pairs $(P,Q)$ that lead to a Lucas pseudoprime for a composite $n$ are characterized by the following formula:
\begin{theorem}[Baillie et al. \cite{Lucas-Baillie}]\label{countinglpsp} Let $D$ be a fixed positive integer, and let $n= \prod_{i=1}^s p_i^{r_i}$ be a positive odd integer with $\gcd(D,n)=1$. Then, the number of distinct values of $P$ modulo $n$, for which there is a $Q$ such that $P^2-4Q \equiv D \bmod n$ and $n$ is a Lucas pseudoprime for $(P,Q)$ is given by:
$$
L(D,n)=\prod_{i=1}^s (\gcd(n-\epsilon_D(n), p_i-\epsilon_D(p_i))-1).
$$
\end{theorem}

\subsection{The Baillie-PSW Primality Test}\label{baillietest}

The Baillie-PSW primality test combines a single Fermat/ Miller-Rabin test with base 2 with a (strong) Lucas test. 
Let us formally introduce the algorithm.
\begin{algorithm}
\caption{\textsc{BailliePSW}($n$)}
\label{Baillie-PSW-alg}
\textbf{Input:} Odd integer $n$ to test for primality.\\
\textbf{Output:} A probable prime using the Baillie-PSW test or ``Composite''.
\begin{enumerate}
\item If $n$ is divisible by any prime less than some convenient limit, e.g. 1000, \textbf{output} ``Composite'', else continue.
\item If $n$ is not a (strong) pseudoprime base 2, \textbf{output} ``Composite'', else continue.
\item Check if $n$ is not a perfect square and use one of the methods to determine $(P,Q)$:
\begin{itemize}
\item \textit{Method A:} Let $D$ be the first element of the sequence $5, -7, 9, -11, 13, \dots$ with $\epsilon_D(n)=-1$. Let $P=1$ and $Q=(1-D)/4$.
\item \textit{Method B:} Let $D$ be the first element of the sequence $5, 9, 13, 17,\dots$ 
with $\epsilon_D(n)=-1$. Let $P= \min\{m \in \N \mid m \text{ odd, } m>\sqrt{D}\}$and $Q=(P^2-D)/4.$
\end{itemize}
\item If $n$ is not a (strong) Lucas pseudoprime for $(P,Q)$, \textbf{output} ``Composite'', else \textbf{output} $n$. 
\end{enumerate}
\end{algorithm}

To date, no composites have been identified passing a Baillie-PSW test, leading to the conjecture that this test could, in fact, determinstically test primality. Gilchrist \cite{gilchrist} even verified that using  Method A of Algorithm \ref{Baillie-PSW-alg} no Baillie-PSW pseudoprimes below $2^{64}$ exist. In the next section we discuss why the Fermat/ Miller-Rabin and (strong) Lucas test with well-chosen parameters might be independent of each other.

\subsection{Orthogonality in the Baillie-PSW Test}\label{orthogonality}
If we choose the parameters $D$, $P$ and $Q$ as described in method B of Algorithm \ref{Baillie-PSW-alg}, then the first 50 Carmichael numbers and several other base-2 Fermat pseudoprimes will never be Lucas pseudoprimes \cite{Lucas-Baillie}. We now give a heuristic argument for the potential orthogonality of the tests, most of them given by Arnault \cite{doct-thesis-arnualt}. However, to understand this, let us first establish the necessary mathematical prerequisites.

\hspace{-2.8mm} Let $\mathcal{O}$ be the set of algebraic integers and $\mathcal{O}_{\mathbb{Q}[\sqrt{D}]}$ the ring of integers of $\mathbb{Q}[\sqrt{D}].$ 

\begin{lemma}[Arnault \cite{doct-thesis-arnualt}]\label{correspendence,tau}
Let $P, Q$ be integers such that $D=P^2-4Q\neq 0$. Let $n$ be an integer relatively prime to $2QD$.
For the Lucas sequences $(U_n), (V_n)$ associated with $P,Q$ we have tha relations
$$
U_k=\frac{\alpha^k-\beta^k}{\alpha - \beta}, \hspace{1mm} V_k=\alpha^k+\beta^k, \hspace{1mm} \textnormal{ for all } k \in \mathbb{N},$$
where $\alpha, \beta$ are the two roots of the polynomial $X^2-PX+Q.$ We put $\tau=\alpha {\beta}^{-1} \in \mathcal{O}_{\mathbb{Q}[\sqrt{D}]}$. For $k\in\mathbb{N}$, we have the equivalences 
\begin{align*}
    n \mid U_k & \Leftrightarrow \tau^k = 1, \\
    n \mid V_k & \Leftrightarrow \tau^k = -1.
\end{align*}
In particular, if $n$ is composite and relatively prime to $2QD$, it is a $slpsp(P,Q)$ if and only if 
\begin{align*}
    \tau^q \equiv 1 \bmod n,
    \vspace{-3mm}
\end{align*}
or
\begin{align*}
\vspace{-3mm}
    \text{there exists } i \text{ such that } 0 \leq i < \kappa \text{ and } \tau^{2^iq} \equiv -1 \bmod n,
\end{align*}
where $n - \epsilon_D(n) =2^\kappa q$ with $q$ odd.
\end{lemma}

For $\alpha= a+b\sqrt{D}\in \mathbb{Q}[\sqrt{D}]$, we define its \emph{conjugate} $\overline{\alpha}=a-b\sqrt{D}$. The \emph{norm} of $\alpha$ is defined by $N(\alpha)=\alpha \overline{\alpha}.$ We denote the multiplicative group of norm 1 elements of $\mathcal{O}_{\mathbb{Q}[\sqrt{D}]}/(n)$ by $(\mathcal{O}_{\mathbb{Q}[\sqrt{D}]}/(n))^\wedge$.
The following proposition connects $P$ and $Q$ defined by the Lucas sequence with the elements $\tau$ of norm 1 in $\mathcal{O}_{\mathbb{Q}[\sqrt{D}]}$:
\begin{proposition}[Arnault \cite{doct-thesis-arnualt}]\label{count-tau-P}
Let $D$ be an integer that is not a perfect square and let $n>1$ be an odd integer with $\gcd(n,D)=1.$ Then, for every integer $P$, there exists an integer $Q$, uniquely determined modulo $n$, such that $P^2-4Q \equiv D \bmod n$. Furthermore, the set of integers $P$, such that
\begin{equation*}
    \begin{cases}
    0 \leq P < n, \\
    \gcd(P^2-D,n)=1 \text{ i.e. } \gcd(Q,n)=1
    \end{cases}
\end{equation*}
is in one-to-one correspondence with the elements $\tau \in (\mathcal{O}_{\mathbb{Q}[\sqrt{D}]}/(n))^\wedge$, such that  $\tau -1$ is a unit in $\mathcal{O}_{\mathbb{Q}[\sqrt{D}]}/(n).$ 
This correspondence is expressed by the formulas:
\begin{equation*}
    \begin{cases}
    \tau \equiv (P+\sqrt{D})(P-\sqrt{D})^{-1}\\
    P \equiv \sqrt{D}(\tau +1) (\tau -1)^{-1} 
    \end{cases} \bmod n\mathcal{O}_{\mathbb{Q}[\sqrt{D}]}.
\end{equation*}
\end{proposition}

The next lemma will be used in the heuristic argument:

\begin{lemma}\label{gcd-lemma}
Let $n=p_i m_i$ be an odd integer. Then the following holds
\begin{align*}
\gcd(n-1,p_i-1)&=\gcd(n-1,m_i-1),\\
\gcd(n+1,p_i+1)&=\gcd(n+1,m_i-1).
\end{align*}
\end{lemma}
\begin{proof} For all $a,b$ we have that $\gcd(a+b,a)=\gcd(a,b)$ and $\gcd(a-1,a)=1$. With this, we obtain for the first equality
\begin{align*}
\gcd(n-1,p_i-1)&=\gcd((m_i-1)p_i+(p_i-1),p_i-1)
=\gcd((m_i-1)p_i,p_i-1)\\
&=\gcd(m_i-1,p_i-1)= \gcd((p_i-1)m_i, m_i-1)\\
&= \gcd((p_i-1)m_i+(m_i-1),m_i-1)=\gcd(n-1,m_i-1)
\end{align*}
The second equality can be established in a similar manner.
\end{proof}

Let us now give the heuristic argument why composite numbers rarely pass the Baillie-PSW test.

Let $P, Q, D$ and $b$ be integers satisfying $P^2-4Q = D$. Consider $n=p_1\dots p_r$ with $\gcd(n,QD)$ and $\legendre{D}{n}=-1$ such that $n$ is both a $psp(b)$ and $lpsp(P,Q)$. Let $\tau$ be the element associated with the pair $(P,Q)$ in Proposition \ref{count-tau-P}. Fermat's Little Theorem implies:
\begin{equation*}
\begin{cases}
b^{\gcd(n-1,p_i-1)} \equiv 1 \text{ in } \mathbb{Z}/n\mathbb{Z},\\
\tau^{\gcd(n+1,p_i-\epsilon_D(p_i)} \equiv 1 \text{ in } \mathcal{O}_{\mathbb{Q}[\sqrt{D}]}/n\mathcal{O}_{\mathbb{Q}[\sqrt{D}]}
\end{cases}
\textnormal{for every } i.
\end{equation*}
We let
\begin{equation*}
    \begin{cases}
    d_i =\gcd(n-1,p_i-1),\\
    d_i'=\gcd(n+1,p_i-\epsilon_D(p_i)).
    \end{cases}
\end{equation*}
We can express the equivalences as follows:
\begin{align*}
    b^{n-1} \equiv 1 \bmod p_i &\Leftrightarrow b \text{ is a }(p_i-1)/{d_i}\text{th root }\bmod p \\
    \tau^{n+1} =1 \text{ in } \mathcal{O}/p_i\mathcal{O} &\Leftrightarrow \tau\text{ is a } (p_i-\epsilon_D(p_i))/{d_i'}\text{th root in } (\mathcal{O}/p_i\mathcal{O})^\wedge.
\end{align*}
Heuristically, these relations have a small chance of being true when the integers $d_i$ and ${d_i}'$ are small relative to the group order of $(\mathbb{Z}/p_i\mathbb{Z})^\times$ and $(\mathcal{O}/p_i\mathcal{O})^\wedge$. 

For $\epsilon_D(p_i)=1$, we observe that
$$
d_id_i'= \gcd(n-1,p_i-1)\gcd(n+1,p_1-1)\leq 2(p_1-1),
$$
implying that it is not possible for both $\gcd$s to be large.

For $\epsilon_D(p_i)=-1$, we let $n=p_im_i$. By Lemma \ref{gcd-lemma}, we have
\begin{equation*}
\begin{cases}
\gcd(n-1,p_i-1)=\gcd(n-1,p_i-1)=\gcd(n-1,m_i-1), \\
\gcd(n+1,p_i+1)=\gcd(n+1,p_i+1)=\gcd(n+1,m_i-1).
\end{cases}
\end{equation*}
We conclude that
$$
d_id_i'=\gcd(n-1,p_i-1)\gcd(n+1,p_i+1)\leq 2(m_i-1),
$$
and, following the same argument as above, it is evident that the $\gcd$s cannot be large. Consequently, $d_i$ and $d_i'$ rather small, leading to a scarcity of pseudoprimes.

\subsection{Rationale for Avoiding \texorpdfstring{$\epsilon_D(n)=1$}{epsilon}}\label{rationale}

We now discuss why it is best to avoid the case where $\epsilon_D(n)=1$, as in such instances, the Fermat/ Miller-Rabin test and the (strong) Lucas test are not independent. Let us consider the various scenarios that can arise.

\subsubsection{When D is a Perfect Square}
If $D\neq 0$ is a perfect square, the strong Lucas test reduces to the Rabin-Miller test. If $\gcd(n,2D)=1$, we put $T=\alpha\beta^{-1} \bmod n$, with $\alpha, \beta \in \mathbb{Z}$. Lemma \ref{correspendence,tau} yields the following equivalences for $k\in \mathbb{N}$:
\begin{align*}
    n \mid U_k & \Leftrightarrow T^k \equiv 1 \bmod n \\
    n \mid V_k & \Leftrightarrow T^k \equiv -1 \bmod n.
\end{align*}
For every $i$, the decompositions $n-1=n-\epsilon_D(n)=2^\kappa q$, $p_i-1=p_i-\epsilon_D(p_i)=2^{k_i}q_i$ are the same, making $n$ a $slpsp(P,Q)$ if and only if it is an $spsp(2QD)$.
There is an easy fix to ensure that $D$ is not a perfect square, namely by applying Newton's method for square roots.

\begin{algorithm}\label{newton}
 \textbf{Input:} $n$-bit integer $D$ to check for being a perfect square.\\
\textbf{Output:} "Perfect square" or "Not a perfect square".
    Set $m=\lceil \frac{n}{2}\rceil$ and $i=0$\;
    Select random $x_0$ such that $2^{m-1} \leq x_0 < 2^{m}$\;

    \While{$x_i^2 \geq 2^m + D$}{
        $i = i + 1$\;  $x_i = \frac{1}{2}(x_{i-1} + \frac{D}{x_{i-1}})$\;
    }
    
    \eIf{$D = \lfloor {x_i}^2 \rfloor$}{
        status = "perfect square"\;
    }{
        status = "not a perfect square"\;
    }
    
    \textbf{Return} status\;
    \caption{Checking for a perfect square using Newton's method}
\end{algorithm}

\vspace{-2mm}

\subsubsection{When D is a Square Modulo n}
Now, let $D$ not necessarily be a perfect square, but a square modulo $n$. The next lemma allows us to construct Lucas pseudoprimes from Fermat pseudoprimes.

\begin{lemma}\label{squaremodn}
Let $n$ be an odd integer that is both a $psp(b)$ and $psp(c)$. Then, $n$ is a $lpsp(P,Q)$ for $P \equiv b + c $ and $Q \equiv bc \bmod n$. 
\end{lemma}
\begin{proof}
Let $\alpha$ and $\beta$ represent the distinct roots of the polynomial $X^2-PX+Q$. Then $\{ \alpha, \beta \} \equiv \{b,c\} \bmod n$, because the quadratic polynomials $X^2-PX+Q$ and $X^2-(b+c)X+bc$ have congruent coefficients modulo $n$. Since $\legendre{D}{n}=\legendre{P^2-4Q}{n}=\legendre{(b-c)^2}{n}=1$, we get $U_{n-1}= \frac{b^{n-1}-c^{n-1}}{b-c} \equiv 1$, so $n$ is a $lpsp(P,Q).$
\end{proof}
Let $P$ and $Q$ satisfy $D=r^2 \bmod n$. Solving the simultaneous equations $P=b+c$ and $Q=bc$ modulo $n$ yields $b=\frac{(P-r)(n+1)}{2}$ and $c=\frac{(P+r)(n+1)}{2}$. For instance, if $n$ is a $psp(2)$, it might be a $psp(b)$ and a $psp(c)$, assuming that $\gcd(n,bc)=1$. This is plausible, considering the possibility of it being a Carmichael number or either $b$ or $c$ being $\pm 1$. Lemma \ref{squaremodn} implies that $n$ is a $lpsp(P,Q)$, indicating that the Lucas test will not be independent of the Fermat test.

\subsubsection{When the Jacobi Symbol is 1}
Most is from \cite{Lucas-Baillie}. Let $n=\prod_{i=1}^s p_i^{r_i}$ and $\epsilon_D(n)=1$. By Theorem \ref{countinglpsp}, there  $\prod_{i=1}^s (\gcd(n-1, p_i - \epsilon_D(p_i))-1)$ values of values $P \mod n$, for which there exists a $Q$ making $n$ a $lpsp(P,Q)$. Similarly, $n$ is by Theorem \ref{number-Fermat} a $psp(a)$ for $\prod_{i=1}^s \gcd(n-1, p_i-1)$ values of $a \mod n$. The product $\gcd(n-1, p_i-1)\cdot \gcd(n-1, p_i+1)$ is less than $2(p_1+1)$, but for $\epsilon_D(p_i)=+1$, the $\gcd$s are equal, allowing both to be large. In many cases, if $n$ is a $lpsp(P,Q)$ for many values of $P$ with $\epsilon_D(n)= +1$, then might also be a $psp(a)$ for many values of $a$. The computer calculations support these observations \cite{Lucas-Baillie}.
\hspace{-1mm} Let us do a similar analysis for the strong Lucas test. Again, let $n=\prod p_i^{r_i}$ and  $\epsilon_D(n)=1$. Let $n-1=2^{\kappa}q$, $p_i- \epsilon_D(p_i)=2^{k_i}q_i$, and $p_i-1=2^{l_i}s_i$ with $q_i, s_i$ odd. The number of pairs $(P,Q)$ such that $0 \leq P, Q <n$, $D=P^2-4Q$, $\gcd(Q,n)=1$, making $n$ a $slpsp(P,Q)$ is, by Theorem \ref{SL(D,n)}, equal to $\prod_{i=1}^s (\gcd(q,q_i)-1) + \sum_{j=0}^{k_1-1}2^{js}\prod_{i=1}^s\gcd(q,p_i-\epsilon_D(p_i))$. The number of bases $a$ such that $n$ is a $spsp(a)$ is, by \cite{Monier} and \cite{Rabin} equal to $(1+\sum_{j=0}^{k_1-1}2^{js}) \prod_{i=1}^s \gcd(q,p_i-1)$. Again, the product $\gcd(q,p_i-1)\cdot\gcd(q,p_i+1) $ cannot exceed $2\gcd(q,p_i+1)$, but for $\epsilon_D(p_i)=+1$, the $\gcd$s are the same, allowing both to be large.

\section*{Acknowledgments}
Supported by the 6G research cluster funded by the Federal Ministry of Education and Research (BMBF) in the programme of  ``Souverän. Digital. Vernetzt.'' Joint project 6G-RIC, project identification number: 16KISK020K.


\begin{thebibliography}{99}
\bibitem{gilchrist}
Jeff Gilchrist.
\textit{Pseudoprime Enumeration with Probabilistic Primality Tests}.
Available online: \url{http://gilchrist.great-site.net/jeff/factoring/pseudoprimes.html?i=1}.

\bibitem{Pomerance-Fermat1}
C. Pomerance.
\textit{A new lower bound for the pseudoprime counting function}.
\textit{J. Math.}, 26, 4-9, 1982.

\bibitem{Pompseudo}
C. Pomerance, J. L. Selfridge, and S. S. Wagstaff, Jr.
\textit{The Pseudoprimes to $25\cdot 10^9$}.
\textit{Mathematics of Computation}, 35, pp 1003-1026, 1980.

\bibitem{Pomerance-Fermat2}
C. Pomerance.
\textit{On the Distribution of Pseudoprimes}.
\textit{Math. Comp.}, 37, 587-593, 1981.

\bibitem{alford1994there}
William R Alford, Andrew Granville, and Carl Pomerance.
\textit{There are infinitely many Carmichael numbers}.
\textit{Annals of Mathematics}, 139(3), 703--722, 1994.

\bibitem{pomerance1984there}
Carl Pomerance.
\textit{Are there Counter-Examples to the Baillie-PSW Primality Test}.
\textit{Dopo Le Parole aangeboden aan Dr. AK Lenstra. Privately published Amsterdam}, 1984.

\bibitem{primesinP}
Manindra Agrawal, Neeraj Kayal, and Nitin Saxena.
\textit{PRIMES is in P}.
\textit{Annals of Mathematics}, pages 781--793, 2004.

\bibitem{Lucas-Baillie}
Robert Baillie and Samuel S Wagstaff.
\textit{Lucas Pseudoprimes}.
\textit{Mathematics of Computation}, 35(152), 1391--1417, 1980.

\bibitem{DamEtAl}
Ivan Damg{\aa}rd, Peter Landrock, and Carl Pomerance.
\textit{Average Case Error Estimates for the Strong Probable Prime Test}.
\textit{Mathematics of Computation}, 61(203), 177--194, 1993.

\bibitem{doct-thesis-arnualt}
François Arnault.
\textit{Sur Quelques Tests Probabilistes de Primalité}.
PhD thesis, Poitiers, 1993.

\bibitem{Rabin}
Michael O Rabin.
\textit{Probabilistic Algorithm for Primality Testing}.
\textit{Journal of Number Theory}, 12, 128--138, 1980.

\bibitem{Rabin-Mon-Lucas}
François Arnault.
\textit{The Rabin-Monier Theorem for Lucas Pseudoprimes}.
\textit{Mathematics of Computation}, 66(218), 869--881, 1997.

\bibitem{Burthe}
Ronald Burthe Jr.
\textit{Further Investigations with the Strong Probable Prime Test}.
\textit{Mathematics of Computation}, 65(213), 373--381, 1996.

\bibitem{incrementalBrandt}
Jørgen Brandt and Ivan Damg{\aa}rd.
\textit{On Generation of Probable Primes by Incremental Search}.
In \textit{Advances in Cryptology—CRYPTO’92: 12th Annual International Cryptology Conference Santa Barbara, California, USA August 16--20, 1992 Proceedings 12}, pages 358--370, Springer, 1993.

\bibitem{Monier}
Louis Monier.
\textit{Evaluation and Comparison of Two Efficient Probabilistic Primality Testing Algorithms}.
\textit{Theoretical Computer Science}, 12(1), 97--108, 1980.

\bibitem{Einseleav}
Semira Einsele and Kenneth Paterson.
\textit{Average Case Error Estimates of the Strong Lucas Test}.
\textit{Designs, Codes and Cryptography}, 2024, 1--38, Springer.

\bibitem{williams}
Hugh C Williams.
\textit{On Numbers Analogous to the Carmichael Numbers}.
\textit{Canadian Mathematical Bulletin}, 20(1), 133--143, 1977.
\end{thebibliography}
\end{document}